\documentclass[a4paper]{article}      

\usepackage{amsmath,amssymb,amsfonts,amsthm} 
\usepackage[mathcal]{eucal}
\usepackage{cancel}
\usepackage{graphicx}
\usepackage{authblk}
\usepackage{tikz}

\usepackage[margin=1cm,font=small]{caption}

\usepackage{layout} 

\newtheorem{Theorem}{Theorem}
\newtheorem{Proposition}{Proposition}

\newtheorem{Lemma}{Lemma}

\newcommand{\C}{\mathbb{C}}
\newcommand{\R}{\mathbb{R}}

\newcommand{\Z}{\mathbb{Z}}
\newcommand{\N}{\mathbb{N}}
\newcommand{\T}{\mathbb{T}}

\newcommand{\dd}{\mathrm{d}}

\newcommand{\ed}{\mathrm{e}}

\newcommand{\supp}{\mathrm{supp}}

\newcommand{\beq}{\begin{equation}}  
\newcommand{\eeq}{\end{equation}}  

\newcommand{\str}{ |}

   \let\ep=\epsilon

  \let\om=\omega

\newcommand{\caN}{{\mathcal N}}

\newcommand{\Id}{{\mathrm{Id}}}
\newcommand{\spanv}{{\mathrm{span}}}

\newcommand{\modek}{{\mathrm k}}

\newcommand{{\underj}}{\underline{j}}

\begin{document}

\title{
Asymptotic localization of energy in \\ non-disordered oscillator chains}

\author[1]{Wojciech De Roeck} 
\affil[1]{\footnotesize{Instituut voor Theoretische Fysica, KU Leuven, Belgium, \texttt{Wojciech.DeRoeck@fys.kuleuven.be}}}

\author[2]{Fran\c cois Huveneers}
\affil[2]{\footnotesize{CEREMADE, Universit\' e Paris-Dauphine, France, \texttt{huveneers@ceremade.dauphine.fr}}}

\date{\today}

\maketitle

\begin{abstract}
We study two popular one-dimensional chains of classical anharmonic oscillators: the rotor chain and a version of the discrete non-linear Schr\"odinger chain.  
We assume that the interaction between neighboring oscillators, controlled by the parameter $\epsilon >0$, is small.
We rigorously establish that the thermal conductivity of the chains has a non-perturbative origin, with respect to the coupling constant $\epsilon$, 
and we provide strong evidence that it decays faster than any power law in $\epsilon$ as $\epsilon \rightarrow 0$. 
The weak coupling regime also translates into a high temperature regime, suggesting that the conductivity vanishes faster than any power of the inverse temperature. 
To our knowledge, it is the first time that a clear connection is established between KAM-like phenomena and thermal conductivity.
\end{abstract}

\section{Introduction}

The rigorous derivation of transport properties of solids from molecular dynamics is a big and inspiring challenge in statistical mechanics out of equilibrium. 
It has been recognized since a long time that the transfer of energy could be strongly reduced, or even suppressed, in some Hamiltonian systems.
Anderson localization provides probably the  clearest example of this phenomenon. 
In the context of thermal transport, it is realized in disordered harmonic crystals \cite{ber}\cite{nac}, which constitute however a very untypical class of solids, 
since they are equivalent to an ideal gas of non-interacting phonons. 
For interacting quantum systems, one expects that the phenomenon of Anderson localization can persist  in some regimes, giving rise to the so-called `many-body localization' \cite{altschuler}. 
Recently, a mathematical approach to this question was developed in \cite{imbriespencer}.  
We learned of this work shortly after starting the present project, and it was a source of inspiration for us, especially for the perturbative part in Section \ref{sec: approximate change of variables}.

At finite volume, Nekhoroshev estimates \cite{pos 1999} and the KAM theorem \cite{pos 2001}
furnish a whole class of classical Hamiltonians that allow energy to be spread through the system only at very slow rates for all initial condition,
and not at all for some of them. 
These results partially extend to finite energy excitations of Hamiltonians depending on infinitely many variables (see \cite{ben}\cite{fro}\cite{pos 1990} among others).
At infinite volume, time-periodic and spatially localized solutions, called breathers, are also known to exist for generic type of classical oscillators chains \cite{mac}.  
In \cite{Maiocchi Carati}, an extensive adiabatic invariant is shown to exist for the Klein Gordon lattice at low but positive temperature 
(i.e.\ for initial conditions of infinite energy), implying an asymptotically slow mixing rate of the system.
As such however, all these results are of little help to understand the thermal conductivity of solids.

In this paper, we analyze two classical chains of strongly anharmonic oscillators (without disorder), 
and we show asymptotic localization of energy in a regime characterized by high thermal fluctuations, in comparison with the coupling strength between near atoms.
Let $\epsilon > 0$ denote a parameter controlling the strength of the coupling. 
We establish that energy can only diffuse through these systems at times that are larger than any inverse power of $\epsilon$ as $\epsilon \rightarrow 0$, 
except perhaps for a set of states, whose probability is itself smaller than any inverse power in $\epsilon$, with respect to the Gibbs state at a positive temperature $T$. 
In that sense, our result could be thought of as an analog of Nekhoroshev estimates, at infinite volume and positive temperature.
We hope that these results also provide some complementary view on the slow relaxation to equilibrium observed for chains with strong anharmonic on-site pinning \cite{hai}.

The first system we consider is a chain of rotors, 
consisting of particles constrained to move on circles, and weakly coupled through cosine interactions.  
Numerical studies indicate that this chain behaves as a normal conductor \cite{gia liv}, though the conductivity becomes divergent as temperature is sent to zero. 
The defocusing discrete non-linear Schr\"odinger chain is the second system we look at.
We study this chain in the regime where the on-site anharmonic pining dominates the weak harmonic coupling.   
Here as well, simulations show this chain to be a normal conductor \cite{iub}. 
It is known that, besides energy, these chains preserve a second quantity (see Section \ref{sec: models and results} below). 
To stress that our results do not depend on this, we allow for an extra interaction, that breaks the second conservation law. 

Our results ultimately rest on a phenomenon that is at the heart of all the results above:  
close individual atoms typically oscillate each at different frequencies, so that resonances, that are responsible for energy transfer in a perturbative regime, are rarely observed. 
To explain this a bit further, we find it useful to introduce a comparison with weakly coupled disordered chains. 
Assuming there the on-site potential to be harmonic, uncoupled atoms simply oscillate at a fixed but random eigenfrequency.
When a small interaction is turned on, only a few disconnected resonant spots are created here and there, 
corresponding to places where the eigenfrequencies of near atoms are, in very good approximation, in specific ratios with respect to each others. 
This observation has allowed to conclude to asymptotic localization of energy for a wide class of interaction potentials \cite{huv}, 
and, for harmonic interactions, to a true localization \cite{fro spe} as long as the coupling is not too large.  

Let us now move back to our non-disordered chains. 
Since the on-site interaction is strongly anharmonic, each uncoupled atom oscillates at a frequency that depends on its energy.  
Moreover, in the absence of interaction, the Gibbs state is a product measure, so that the eigenfrequency of each oscillator is here as well chosen randomly.
So far, the comparison with inhomogeneous chains is thus perfect. 
When the interaction is turned on, it is still so that rare resonant spots will appear here and there. 
However, these resonant islands are no longer attached to a fixed place. 
Instead, as a bit of energy get transferred, the eigenfrequencies are slightly modified, so that resonant sites can be destroyed here and recreated there. 
This phenomenon a priori favors the transport of energy. 
In fact, our main difficulty compared to \cite{huv}, was to show that this process itself occurs so slowly that it is irrelevant at the time scales we consider.   
Once this difficulty is overcome, the results of this paper resemble very closely the analogous statements in \cite{huv}.  
Those results were in turn inspired by \cite{liv} where the weak coupling limit for oscillator chains with energy-conserving dynamics was analyzed rigorously for the first time. 

The paper is organized as follows. 
Our results are stated in Section \ref{sec: models and results}.
The rest of the paper is devoted to the proofs.
We have not been able to handle both chains in a unified way, though it is mainly a question of details. 
As a consequence, Sections \ref{sec: approximate change of variables} to \ref{sec: proof of the principal theorem} exclusively deal with the rotor chain, 
while the non-linear Schr\"odinger chain is considered in Section \ref{sec: adaptation to nls}.

In Section \ref{sec: approximate change of variables}, 
a KAM-like change of variables is constructed, that isolates from the rest the part of the interaction giving rise to resonances. 
The stability of resonant islands is studied in Section \ref{sec: resonant frequencies}.
Our main result is finally shown in Section \ref{sec: proof of the principal theorem} for the rotor chain. 
Adaptations needed to handle the non-linear Schr\"odinger chain are explained in Section \ref{sec: adaptation to nls}. 
The final Section \ref{sec: last section}  contains the proof of three corollaries.

\section{Models and Results}\label{sec: models and results}

We define precisely the chains under study as well as the thermal conductivity, 
and state our results together with some comments.

\subsection{Models}

Let $N\ge 1$ be an odd integer and let $\Z_N = \{ -(N-1)/2, \dots , (N-1)/2 \}$. 
Let also $\gamma \ge 0$.
For $\gamma = 0$, the two dynamics we study preserve both the total energy and a second quantity: 
the momentum for the rotor chain, and the $\ell^2$-norm for the discrete non-linear Schr\"odinger chain. 
When $\gamma > 0$, these extra conservation laws are broken, so that only energy remains conserved. 
We will assume free boundary conditions, though that is only a matter of convenience: 
all our conclusions would still hold for other choices of boundary conditions.

\paragraph{The rotor chain.} 
The phase space consists of the points 
\begin{equation*}
(q,\omega) 
\; = \;
(q_x,\omega_x)_{x \in \Z_N} 
\; \in \; 
\Omega  = \Omega_N =   (\T \times \R)^{N} \quad \text{with} \quad \T = \R/(2\pi\Z).
\end{equation*}
The Hamiltonian is
\begin{align}
H (q,\omega) 
\; &= \;
D(\omega) + \epsilon V(q,\omega)
\; = \; \sum_{x\in\Z_N} H_x (q,\omega)
\; = \; \sum_{x\in \Z_N} \big( D_x (\omega) + \epsilon V_x (q) \big)
\nonumber\\
\; &= \; 
\frac{1}{2} \sum_{x\in\Z_N} \omega_x^2
\; + \;
\epsilon \sum_{x\in\Z_N} 
\Big(
 \gamma  (1 - \cos q_x) +  \big(1 - \cos (q_x - q_{x+1}) \big)  
\Big) ,
\label{Hamiltonian Rotors}
\end{align}
with the convention $q_{(N+1)/2} = q_{(N-1)/2}$, so that free boundary conditions are imposed on both sides.
The Hamilton equations of motion are
\begin{equation}\label{Hamilton equations}
\dot{q} \; = \; \nabla_\omega H 
\qquad \text{and} \qquad 
\dot{\omega} \; = \; - \nabla_q H. 
\end{equation}
The total momentum $\sum_{x} \omega_x$ is a conserved quantity only at $\gamma = 0$.
Given an initial condition $(q,\omega) \in \Omega$, we denote the Hamiltonian flow by $ (X_\epsilon^t (q,\omega))_{t\geq 0} \subset \Omega$.

\paragraph{The discrete non-linear Schr\"odinger chain.}
The phase space consists of the points 
\begin{equation*}
\psi \; = \; (\psi_x)_{x\in\Z_N} 
\; \in \; \Omega  = \Omega_N =  \C^{N} \simeq (\R^2)^N. 
\end{equation*}
The Hamiltonian is
\begin{align}
H (\psi ) 
\; &= \; 
D (\psi ) + \epsilon V (\psi)
\; = \; 
\sum_{x\in \Z_N} H_x (\psi)
\; = \; \sum_{x \in \Z_N} 
\big( 
D_x(\psi) + \epsilon V_x (\psi)
\big)
\nonumber\\
\; &= \; 
\frac{1}{2}\sum_{x\in \Z_N} |\psi_x|^4
\; + \; \epsilon
\sum_{x\in\Z_N}
\Big(
\gamma (\psi_x + \overline{\psi}_x)^2 + |\psi_x - \psi_{x+1}|^2
\Big),
\label{Hamiltonian DNLS}
\end{align}
with again the convention $\psi_{(N+1)/2} = \psi_{(N-1)/2}$.
Writing  $H(\psi)$ as $H(\psi,\overline{\psi})$, the Hamilton equations of motion take the redundant form
\begin{equation*}
i \dot{\psi} \; = \; \nabla_{\overline{\psi}} H  
\qquad \text{and} \qquad
i \dot{\overline{\psi}} \; = \; - \nabla_{\psi} H .
\end{equation*}
The total $\ell^2$-norm $\sum_x |\psi_x|^2$ is a conserved quantity only at $\gamma = 0$.
Given an initial condition $\psi \in \Omega$, we denote the Hamiltonian flow by $ (X_\epsilon^t (\psi))_{t \geq 0} \subset \Omega$.

To see the analogy between this chain and the rotor chain, we could move to action-angle, or polar, coordinates. 
Writing 
\begin{equation*}
\omega_x \; = \; |\psi_x|^2
\qquad \text{and} \qquad 
\tan q_x \; = \;  \Im \psi_x/ \Re \psi_x,
\end{equation*}
the Hamiltonian \eqref{Hamiltonian DNLS} is recast as
\begin{equation*}
H (q,\omega) 
\; = \; 
\sum_{x\in\Z_N}
\Big(
\frac{\omega_x^2}{2} +4 \gamma \epsilon\omega_x \cos^2 {q}_x + \epsilon \big(\omega_x + \omega_{x+1} - 2\sqrt{\omega_x \omega_{x+1}} \cos (q_x - q_{x+1})\big)
\Big)
\end{equation*}
while Hamilton equations now precisely take the form \eqref{Hamilton equations}.
Unfortunately, this change of variable is not well defined if some frequency $\omega_x$ vanishes, 
implying that the field $\nabla_\omega H$ becomes singular as $\omega_x\rightarrow 0$.
We will, for this reason, not make explicitly use of it.

\subsection{Heat current and thermal conductivity}

Given two functions $f,g\in \mathcal C^\infty (\Omega)$, we define, for rotors, 
\begin{equation}\label{Poisson bracket for rotors}
L_f g 
\; = \; 
\{ f , g \} 
\; = \; 
\nabla_\omega f \cdot \nabla_q g - \nabla_q f \cdot \nabla_\omega g
\; = \; 
- \{ g,f \}
\; = \; 
- L_g f,
\end{equation}
and for the non-linear Schr\"odinger chain,
\begin{equation}\label{Poisson bracket for DNLS}
L_f g 
\; = \; 
\{ f , g \} 
\; = \; 
- i \big( \nabla_{\overline \psi} f \cdot \nabla_{\psi} g - \nabla_{\psi} f \cdot \nabla_{\overline \psi} g \big).
\end{equation}
Given $a \in \Z_N$, we define the energy current $\epsilon J_{a,a+1}$ across the bond $(a,a+1)$ by 
\begin{equation} \label{current}
\epsilon J_{a,a+1} \; = \; L_H \sum_{x > a} H_x
\; = \; 
\bigg\{ \sum_{y \le a} H_y \, ,\, \sum_{x > a} H_x \bigg\}
\; = \; 
\{ H_a , H_{a+1} \}.
\end{equation}
We then define the total, normalized, current $\epsilon \mathcal J$  by 
\begin{equation*}
\epsilon \mathcal J \; = \; \frac{\epsilon}{N^{1/2}} \sum_{a\in \Z_N} J_{a,a+1}.
\end{equation*}

Let $T > 0$ be some fixed temperature. 
The Gibbs state is a measure on $\Omega$ defined, for the rotor chain, by 
\begin{equation*}
f \; \mapsto \; \langle f \rangle_T \; = \; \frac{1}{Z(T)} \int_\Omega f (q,\omega) \, \ed^{- H(q,\omega)/T} \, \dd q \dd \omega, 
\end{equation*}
where $Z (T)$ is a normalization factor such that this measure is a probability measure. 
For the non-linear Schr\"odinger chain, the expression is analogous: $H(q,\omega)$ is replaced by $H(\psi)$, and $\dd q \dd \omega$ is replaced by $\dd \Re (\psi) \dd \Im (\psi)$.
The Green-Kubo conductivity of the system is defined, if the limits exist, as a space-time variance \cite{lep}: 
\begin{equation}\label{conductivity}
\kappa (T,\epsilon) \; = \; \lim_{t\rightarrow \infty} \lim_{N\rightarrow \infty} 
\frac{1}{T^2}
\Big\langle\Big( 
\frac{\epsilon}{\sqrt{t}} \int_0^t \mathcal J_N (X^s_\epsilon) \, \dd s
\Big)^2\Big\rangle_T
\end{equation} 
where we have written $\mathcal J_N$ instead of $\mathcal J$, to remind ourselves that this quantity depends on $N$, and we have set $\mathcal J_N (X^s_\epsilon)=\mathcal J_N \circ X^s_\epsilon$.
We note that, thanks to good decorrelation properties of the Gibbs measure (see Section \ref{sec: last section}), 
the limit $N\rightarrow \infty$ is independent of the boundary conditions.

\subsection{Results}

We start by an abstract result expressing that, in all orders in perturbation in $\epsilon$, 
only local oscillations of the energy field (and hence no persistent currents) can be produced by the dynamics.
\begin{Theorem}\label{the: decomposition fo the current}
Let the Hamiltonian be given by \eqref{Hamiltonian Rotors} or \eqref{Hamiltonian DNLS}. 
Let $T > 0$ be fixed.
Choose any $n \ge 1$ and let then $\mathrm C_n < + \infty$ be large enough. 
For any  $N \ge 1$ and $a \in \Z_N$,
the current across the bond $(a,a+1)$ can be decomposed as 
\begin{equation*}
\epsilon J_{a,a+1} \; = \; L_H U_{a} + \epsilon^{n+1} G_{a}
\end{equation*}
The functions $U_{a}$ and $G_{a}$ are 
smooth, 
of zero average,   $\langle U_a \rangle_T=\langle G_a \rangle_T=0$, and they
depend only on variables labeled by $z\in \Z_N$ with $|z - a| \le \mathrm C_n$,
and satisfy the bounds
\begin{equation}\label{bornes sur U et G}
\langle U_a^2 \rangle_T \; \le \; \mathrm C_n\epsilon^{1/4}, 
\;
\langle (\partial_\sharp U_a)^2 \rangle_T \; \le \; \mathrm C_n \epsilon^{-1/4},
\;
\langle G_a^2 \rangle_T \; \le \; \mathrm C_n, 
\; 
\langle (\partial_\sharp G_a)^2 \rangle_T \; \le \; \mathrm C_n
\end{equation}
where $\sharp$ stands for any of the variables.
\end{Theorem}

We deduce two results on the thermal conductivity from this abstract statement. 
The analysis of the conductivity as defined by \eqref{conductivity} is probably out of reach at the present time. 
We can however obtain some conclusion by assuming that the true value of the integral in \eqref{conductivity} is already attained at a time $t$ that grows as some inverse power in $\epsilon$ as $\epsilon \rightarrow 0$.  One can argue (see e.g.\ Chapter 5 of \cite{ollabernardin}) that this is equivalent to exciting the system locally, and observing the relaxation for a time $t$ of this order. Our result is quite similar in spirit to results about weak coupling limits in such systems, e.g.\  \cite{lukkarinenspohn}\cite{dolgopyatliverani}\cite{liv}, where one describes the dynamics in a scaling limit where coupling vanishes but time goes to infinity.  However, in our case, these scaling limits are trivial in the sense that we do not see any transport on the time scales that we study.  We would find it very interesting to push the analysis to longer time scales and to exhibit a non-vanishing contribution to the conductivity.

So first, we follow the dynamics for a time of order $\epsilon^{-n}$, for an arbitrary large $n$, and let $\epsilon \rightarrow 0$. 
We believe the next theorem to be a strong indication that $\kappa (T,\epsilon) = \mathcal O (\epsilon^m)$ for any $m\ge 1$. To establish this rigorously, one would need to 
exchange the limits $t\rightarrow \infty$ and $\epsilon \rightarrow 0$. 
\begin{Theorem}\label{the: weak coupling conductivity}
Let the Hamiltonian be given by \eqref{Hamiltonian Rotors} or \eqref{Hamiltonian DNLS}. 
Let $T > 0$ be fixed.
Let $1 \le m < n$. Then
\begin{equation*}
\lim_{t\rightarrow \infty} \limsup_{\epsilon \rightarrow 0} \limsup_{N\rightarrow \infty} 
\frac{\epsilon^{-m}}{T^2}
\Big\langle\Big( 
\frac{\epsilon}{\sqrt{\epsilon^{-n}t}} \int_0^{\epsilon^{-n}t} \mathcal J_N (X^s_\epsilon) \, \dd s
\Big)^2\Big\rangle_T
\; = \; 
0.
\end{equation*}
\end{Theorem}
One could speculate whether some non-perturbative effects could lead to a breakdown of the conjecture $\kappa (T,\epsilon) = \mathcal O (\epsilon^m)$. 
We cannot exclude this, and in fact we do not even rigorously know whether the chains we consider are normal conductors for some $\epsilon > 0$, that is, whether $\kappa (T,\epsilon) <\infty$. 
It is however commonly believed that, on  sufficiently large time scales, the dynamics of such systems becomes chaotic. 
As in \cite{huv}, we can, for the rotor chain, mimic this hypothetic non-perturbative chaotic behavior by a stochastic noise that conserves energy, 
and that becomes perceptible on very large time scales, namely $\epsilon^{-(n+1)}$, for some arbitrarily large $n$.
We  are then able to show that the conductivity is finite and not larger than $\epsilon^{n}$, so that it can be attributed to the noise. Instead, we do not know what could be the effect of non-perturbative integrable structures, such as solitons traveling ballistically. 

Let us consider the rotor chain. For $n\ge 1$, we let
\begin{equation}\label{noisy generator}
\mathcal L 
\; = \; 
L_H \; + \; \epsilon^{n+1} S
\quad \text{with} \quad 
S u (q, \omega) \; = \; \sum_{x\in \Z_N} \big( u (q, \dots ,-\omega_x , \dots ) - u (q,\omega) \big) 
\end{equation}
be the generator of a Markov process on $\Omega$. 
Let us denote by $(\mathcal X_\epsilon^t(q,\omega))_{t\geq 0}$ the Markov process generated by $\mathcal L$ and started from the point $(q,\omega)$. 
We denote by $\mathsf E$ the expectation with respect to the realizations of the noise $S$.
\begin{Theorem}\label{the: noisy dynamics}
Let the Hamiltonian be given by \eqref{Hamiltonian Rotors}. 
Let $T > 0$ be fixed.
For any $n\ge 1$, it holds that there is $\mathrm C_n < \infty$ such that, for sufficiently small $\ep>0$,
\begin{equation*}
\lim_{t\rightarrow \infty}  \limsup_{N\rightarrow \infty} 
\frac{1}{T^2}
\Big\langle \mathsf E \Big( 
\frac{\epsilon}{\sqrt{t}} \int_0^{t} \mathcal J_N (\mathcal X^s_\epsilon) \, \dd s
\Big)^2\Big\rangle_T
\; \leq \; \mathrm C_n \epsilon^n
\end{equation*}
\end{Theorem}

From Theorem \ref{the: decomposition fo the current}, we can also deduce a statement that mirrors the well-known Nekhorohsev theorem 
for systems  consisting of a finite number of degrees of freedom (see e.g.\  \cite{pos 1999}).
We recall that the latter states that, for all initial conditions, 
the action coordinates of the uncoupled system remain $\mathrm C \epsilon^b$-close to their original value for a time $\ed^{ \mathrm{C} (1/\epsilon)^{a}}$, 
for some $a,b>0$  and with $\epsilon$ the coupling strength. 
We can reproduce this statement for arbitrary polynomial times, rather than exponential ones, and for a set of configurations that has large probability with respect to the Gibbs state.
In \cite{Carati}, a similar result was obtained, limited however to much shorter time scales.
Let $I =\{a_1, a_1+1,\dots, a_2 \} \subset \Z_N$ be a discrete interval, and let $H_I=\sum_{x} H_x$. Then
\begin{Theorem}\label{the: nekoroshev}
Let the Hamiltonian be given by \eqref{Hamiltonian Rotors} or \eqref{Hamiltonian DNLS}. 
Let $T > 0$ be fixed.
For any $n \ge 1$,  there is  $\mathrm C_n <\infty$ such that, for sufficiently small $\epsilon>0$, and for any $I$ as above, 
\begin{equation}
\Big\langle \big( H_I(X^t_\epsilon)- H_I \big)^2 \Big\rangle_T \leq  \mathrm C_n \ep^{1/4}, \qquad \text{for any}\quad  0\leq t \leq \epsilon^{-n}.
\end{equation}
\end{Theorem}

\subsection{Remarks}

\paragraph{Temperature dependence.}
The behavior of the thermal conductivity $\kappa (T,\epsilon)$ defined by \eqref{conductivity} 
as $\epsilon\rightarrow 0$ for fixed $T > 0$, 
is directly connected to its behavior as $T \rightarrow \infty$ for fixed $\epsilon > 0$.
Indeed, assuming that \eqref{conductivity} is well defined, we have as we will see that, for every $\sigma > 0$,
\begin{eqnarray}
\kappa (\epsilon, T) 
\; &=& \;
\frac{1}{\sigma} \, \kappa (\sigma^2 \epsilon, \sigma^2 T)
\qquad \text{for the rotor chain},
\label{scaling rotors}\\ 
\kappa (\epsilon,T)
\; &=& \;
\frac{1}{\sigma} \, \kappa (\sigma \epsilon, \sigma^2 T)
\qquad \; \,\text{for the non-linear Schr\"odinger chain.}
\label{scaling NLS}
\end{eqnarray}
We therefore also conjecture for the two chains that $\kappa (T,\epsilon \sim 1) = \mathcal O (1/T^m)$ for every $m\ge 1$ as $T\rightarrow \infty$.  As one can check from the calculations below, we obtain also scaling relations like \eqref{scaling rotors}, \eqref{scaling NLS} for the finite-time approximations to the conductivity $\kappa$ that figure in Theorem \ref{the: weak coupling conductivity}, so we could literally restate this result for the high-temperature regime.
An analogous scaling result was obtained in \cite{aok} for a different chain. 

Let us see how to obtain \eqref{scaling rotors}.
Let $\sigma > 0$.
Let us write $H_\epsilon$ instead of $H$ to explicitly keep track of the coupling strength. 
It is computed that, if $(q(t),\omega(t))_{t\ge 0}$ is a solution to Hamilton's equation \eqref{Hamilton equations} for the Hamiltonian $H_\epsilon$ given by \eqref{Hamiltonian Rotors},
then $(q'(t),\omega'(t))_{t\ge 0}$ given by
\begin{equation*}
q'(t) \; = \; q(\sigma t), 
\qquad 
\omega' (t) \; = \; \sigma \, \omega (\sigma t),
\end{equation*}
solves Hamilton's equations for the Hamiltonian $H_{\sigma^2 \epsilon}$.
It is then computed by means of \eqref{current} that $\epsilon J_{a,a+1} = \epsilon \omega_{a+1} \sin (q_a - q_{a+1})$, 
where $\epsilon J_{a,a+1}$ denotes the current through $(a,a+1)$ corresponding to the Hamiltonian $H_\epsilon$. 
Let us denote by $\sigma^2 \epsilon J'_{a,a+1}$ the current corresponding to $H_{\sigma^2 \epsilon}(q',\omega')$.
It holds that  
\begin{equation*}
\frac{\epsilon}{\sqrt t} \int_0^t \mathcal J (q(s),\omega (s)) \, \dd s
\; = \;
\frac{1}{\sigma^{5/2}} \, \cdot \,
\frac{\sigma^2\epsilon}{\sqrt{t'}} \int_0^{t'} \mathcal J' (q'(s'),\omega' (s')) \, \dd s'
\quad \text{with} \quad 
t' = t/\sigma.
\end{equation*}
In the Gibbs measure, the change of variables implies the change $T \mapsto \sigma^2 T$ for the temperature:
\begin{equation*}
\frac{\int u(q,\omega) \ed^{-H_\epsilon(q,\omega)/T} \, \dd q \dd \omega}{\int \ed^{-H_\epsilon(q,\omega)/T} \, \dd q \dd \omega}
\; = \; 
\frac{\int u(q',\omega') \ed^{-H_{\sigma^2\epsilon}(q',\omega')/\sigma^2T} \, \dd q' \dd \omega'}{\int \ed^{-H_{\sigma^2\epsilon}(q',\omega')/{\sigma^2T}} \, \dd q' \dd \omega'}. 
\end{equation*}
The scaling relation \eqref{scaling rotors} then follows from the definition \eqref{conductivity}. 
The scaling relation \eqref{scaling NLS} is obtained analogously: 
it is here observed that, if $(\psi(t))_{t\ge 0}$ is a solution to Hamilton's equations for the hamiltonian $H_\epsilon$ given by \eqref{Hamiltonian DNLS}, 
then $(\psi' (t) = \sqrt \sigma \psi (\sigma t))_{t\ge 0}$ solves Hamilton's equations for the Hamiltonian $H_{\sigma \epsilon}$.

\paragraph{Higher dimensions.}
We conjecture that our results extend to higher dimensional lattices. 
The arguments in Sections \ref{sec: approximate change of variables} and \ref{sec: resonant frequencies} would indeed carry over straightforwardly. 
The evolution of energy appears thus equally frozen for two or three-dimensional lattices as for a one-dimensional one.   
Unfortunately, the proof of Theorem \ref{the: decomposition fo the current} that appears in Section \ref{sec: proof of the principal theorem},  does not extend as such to higher dimensions. 
Although the problem seems to us purely technical and we find it very plausible that one can adapt it to higher dimensions, we have not pursued this here.

\paragraph{Other models.}
Our results depend mainly on three properties of the models:
First, the dynamics of isolated oscillators is one-dimensional, and thus integrable, so that the frequency of oscillation is a well defined concept. 
Second, the isolated oscillators are strongly anharmonic, implying that the frequencies depend on the energy in a non-trivial way. 
Third, the coupling is weak, so that perturbation theory applies. 
It is thus natural to ask whether, for example, our results would also hold for the Hamiltonian 
\begin{equation}\label{quartic chain}
H (q,p) \; = \; \sum_{x\in \Z_N} \Big( \frac{p_x^2}{2} + \frac{q_x^4}{4} + \frac{\epsilon}{2} (q_x - q_{x+1})^2 \Big),
\end{equation}
as it possesses the three listed properties, i.e.\ whether its conductivity is also non-perturbative as $\epsilon \rightarrow 0$.

It turns out that we actually exploit a specific characteristic of the chains that we look at: 
the perturbation only involves a finite number of combinations of the eigenfrequencies of neighboring oscillators, 
meaning technically that we may work with finite trigonometric polynomials (see Section \ref{sec: approximate change of variables}). 
This would not longer be true for the chain defined by \eqref{quartic chain}, 
for which trigonometric polynomials should be replaced by more generic analytic functions. 
While this extra difficulty can be overcome in usual KAM or Nekhoroshev theorems, 
part of our proof would likely break down (see Section \ref{sec: resonant frequencies}).
The generalisation of our theorems to the chain defined by \eqref{quartic chain}  appears thus to us as an open question.

\paragraph{How optimal are our bounds?}
It is numerically observed that the chains under study are normal conductors \cite{gia liv}\cite{iub}, 
so that we expect localization of energy to be at best asymptotic. 
Still, the time scales at which energy starts diffusing could be much larger than any inverse power in $\epsilon$.
At finite volume for example, Nekhoroshev estimates imply the absence of diffusion over exponentially long times. 
However, in \cite{bas}, the thermal conductivity of a classical non-linear disordered chain is studied, 
and it is argued that the localization is broken at a scale that is roughly  of the order of $\ed^{-c\ln^3 (1/\epsilon)}$. 
Since we expect the energy to travel more easily in the non-disordered chains thanks to the mobility of resonants spots, 
we conjecture that, here as well, localization does not persist on longer times than that.     
In other words, we do not think that one can obtain Nekhoroshev estimates in infinite volume for times as long as those in finite volume.

\section{Approximate change of variables}\label{sec: approximate change of variables}

We introduce an auxiliary Hamiltonian $\widetilde H = \widetilde{H}_{n_1}$, 
defined for an arbitrary $n_1\ge1$, and give the needed links between $\widetilde{H}$ and the original Hamiltonian $H$. 
We first introduce some definitions, then state the results, and finally prove them. 
The formulas introduced in the second part are probably best demystified by first reading the beginning of the proof. 
It is seen there that we define a KAM-like formal change of variable.
In contrast to the KAM-scheme however, our expansion is only perturbative, and does not involve any renormalization of the energy of individual atoms at each step.

\subsection{Preliminary definitions}
Throughout all this work, we will deal with functions $f$ in a subspace $\mathcal S (\Omega)$ of $\mathcal C^{\infty} (\Omega)$. 
A function $f$ belongs to $\mathcal S (\Omega)$ if the three following conditions are realized for some number
\begin{equation}\label{parameter r for functions in S}
r = r(f) > 0.
\end{equation}
\begin{enumerate}
\item The function $f$ is a sum of local terms:
\begin{equation}\label{f sum of local terms}
f \; = \; \sum_{x\in \Z_N} f_x
\qquad \text{with} \qquad
\frac{\partial f_x}{\partial q_y} \; = \;  \frac{\partial f_x}{\partial \omega_y} \; = \;  0
\qquad \text{if} \qquad 
|x-y| > r(f).
\end{equation}
This decomposition is not unique.
\item The function $f$ depends on the variable $q$ through a finite number of Fourier modes only:
\begin{equation}\label{f finite number of Fourier modes}
f(q,\omega) 
\; = \;
\sum_{\modek \in \Z^{N}} \widehat{f} (\modek,\omega) \ed^{i \modek \cdot q}
\quad \text{with} \quad
\widehat{f} (\modek,\omega) \; = \; 0 
\;\; \text{if} \;\; \max_x |\modek_x| \;\ge\; r(f).
\end{equation}
As a consequence of the spatial locality in $1.$, it also holds $\widehat{f} (\modek,\omega) = 0$ as soon as $\supp (\modek)$ cannot be included in a ball of radius $r$,  
where $\supp (\modek) = \{ x \in \Z_N : \modek_x \ne 0 \}$.
\item 
Given any $m\ge 1$, and given any differential operator $\mathrm D$, 
with either $\mathrm D = \mathrm{Id}$ or $\mathrm D = \partial_{\sharp_1} \dots \partial_{\sharp_m}$, 
where $\sharp_k$ stands for any of the variables, 
there is a polynomial $p_{\mathrm D}$ on $\R^{2r+1}$ so that, for every $x\in \Z_N$, and $(q,\omega) \in \Omega$,
\begin{equation} \label{polynomial bounds}
| \mathrm D f_x(q,\om) | \; \leq \;  | p_{\mathrm D}(\omega_{x-r}, \dots , \omega_{x+r}) |.
\end{equation}
\end{enumerate}
In the sequel, the dependence on the length $N$ will always be written explicitly;
the parameter $r$ in \eqref{parameter r for functions in S} will never depend on $N$.

Let $\rho \in \mathcal C^\infty ( \R , [0,1]) $ be a smooth cut-off function: 
$\rho (-x) = \rho (x)$ for every $x\in\R$,
$\rho (x) = 1$ for every $x\in [-1,1]$ and $\rho (x) = 0$ for every $x\notin [-2,2]$.
For any $a>0$, we define also $\rho_a$ by $\rho_a (x) = \rho (x/a)$.

Let $0 < \delta < 1$. 
In this section, we assume this number to be independent of $\epsilon$. 
We define an operator $\mathcal R$ on $\mathcal S (\Omega)$ that acts as
\begin{equation*}
(\mathcal R f ) (q,\omega) \; = \; \sum_{\modek \in \Z^{N}} \rho_\delta (\modek \cdot \omega) \widehat{f} (\modek,\omega) \ed^{i \modek \cdot q} .
\end{equation*}
Let $D$ be the function defined in \eqref{Hamiltonian Rotors}. 
Given $f\in \mathcal S (\Omega)$, the equation 
\begin{equation*}
L_D u \; = \; (\Id - \mathcal R)f, 
\end{equation*}
where $L_D = \{ D , \cdot \}$ is defined in \eqref{Poisson bracket for rotors},
can be solved in $\mathcal S (\Omega)$. 
A solution $u$ is given by
\begin{equation*}
u(q,\omega) \; = \; \sum_{\modek \in \Z^N} \frac{1- \rho_\delta (\modek\cdot \omega)}{i \, \modek\cdot \omega} \widehat{f} (\modek,\omega) \ed^{i \modek \cdot q} ,
\end{equation*}
where the sum only goes over terms for which $\modek\cdot \omega \ne 0$.
This is the only solution such that $\widehat{u} (0,\omega) = 0$ for all $\omega \in \R^{N}$ ; 
we will refer to it as the solution to the equation $L_D u  =  (\Id - \mathcal R)f$.

Finally, we will find it convenient to work with formal power series in $\epsilon$:
given a vector space $E$, these are expressions of the form $Y = \sum_{k\ge 0} \epsilon^k Y^{(k)}$, 
where $Y^{(k)} \in E$  for every $k \ge 0$.   
We naturally extend algebraic operations in $E$ to operations between formal series. 
Given $l\ge 0$ and given a formal series $Y$, we define the truncation 
\begin{equation}\label{troncature}
\mathcal T_l (Y) \; = \;  \sum_{k=0}^l \epsilon^k Y^{(k)} \; \in \; E. 
\end{equation}
If a formal power series $Y$ is such that $Y^{(k)} = 0$ for all $k> l$ for some $l \in \N$, 
we will allow ourselves to identify $Y$ with its truncation $\mathcal T_l (Y) \in E$.

\subsection{Statement of the results}\label{subsection: statement of result in perturbative part}
Given $k \ge 1$, let $\pi(k) \subset \mathbb{N}^k$ be the collection of $k$-tupels $\underline{j}=(j_l)_{l=1,\ldots,k}$ of nonnegative integers satisfying the constraint
$$
\sum_{l=1}^k   l j_l =k.
$$
In particular $0 \le j_l \le k$. 

For $k\ge 0$, we recursively define operators $Q^{(k)},R^{(k)}$ and $S^{(k)}$ on $\mathcal S (\Omega)$, as well as functions $U^{(k)}\in\mathcal S (\Omega)$. 
Here and below, let us adopt the convention $A^0= \mathrm{Id}$ for an operator $A$. 
We first set $Q^{(0)} = R^{(0)} = \Id$, $S^{(0)} = 0$ and $U^{(0)} = 0$.
Next, for $k\ge 1$, we define $U^{(k)}$ as the solution to the equation
\begin{equation}\label{definition U k}
L_D U^{(k)} \; = \; (\Id - \mathcal R) \big( S^{(k-1)} D + Q^{(k-1)} V \big).
\end{equation}
and then set
\begin{align}
Q^{(k)} \; &= \; 
\sum_{\underj \in \pi(k)}
\frac{1}{j_1! \ldots j_k!} 
L_{U^{(k)}}^{j_{k}}  \dots L_{U^{(1)}}^{j_{1}}, 
\label{definition Q k}\\
R^{(k)} \; &= \; 
\sum_{\underline{j} \in \pi(k)}
\frac{ (-1)^{j_1 + \dots + j_k}}{j_1! \ldots j_k!} 
L_{U^{(1)}}^{j_{1}} \dots L_{U^{(k)}}^{j_{k}},  
\label{definition R k}\\
S^{(k)} \; &= \; 
\sum_{\substack{\underline{j} \in \pi(k+1):\\ j_{k+1}=0}}
\frac{1}{j_1! \ldots j_k!} 
L_{U^{(k)}}^{j_{k}}  \dots L_{U^{(1)}}^{j_{1}} .
\label{definition S k}
\end{align}

For $n_1\ge 1$, we define 
\begin{equation}\label{definition of H tilde}
\widetilde{H} 
\; = \; 
\widetilde{H}_{n_1} 
\; = \; 
D + \sum_{k=1}^{n_1} \epsilon^k \mathcal R (S^{(k-1)} D + Q^{(k-1)}V).
\end{equation}
The following Proposition is shown in Subsection \ref{sub: proofs change of variable} below.
\begin{Proposition}\label{pro: change of variable}
Let us consider the formal series $R = \sum_{k\ge 0} \epsilon^k R^{(k)}$ of operators on $\mathcal S (\Omega)$. 
\begin{enumerate}
\item
$H \; = \; \mathcal T_{n_1} \big( R \widetilde{H}_{n_1} \big)$. 
\item
For every $f = \sum_{k=0}^{n_1} \epsilon^k f^{(k)}$, it holds that 
\begin{equation*}
L_H \big( \mathcal T_{n_1} (R f) \big)  
\: = \; 
\mathcal T_{n_1} \big( R L_{\widetilde{H}_{n_1}} f \big) 
\; + \; 
\epsilon^{n_1+1} L_V\sum_{k=0}^{n_1}
R^{(n_1-k)} f^{(k)}.
\end{equation*}
\item
The function $\widetilde{H}$ is symmetric under the exchange $\omega \mapsto - \omega$.
Moreover, for any $k \ge 0$, the operator $R^{(k)}$ maps symmetric functions with respect to this operation, to symmetric functions.
\end{enumerate}
\end{Proposition}

The function $\widetilde{H}$ and the formal operator $R$ have several characteristic that are good to remember. 
\begin{enumerate}
\item
Both $\widetilde H$ and $R$ are expressed as power series in $\epsilon$, 
as is seen from \eqref{definition of H tilde} and from the definition of $R$ given in Proposition \ref{pro: change of variable}.
We introduce also the notation 
\begin{equation*}
\widetilde{H} = \sum_{k=0}^{n_1} \epsilon^k \widetilde{H}^{(k)}
\quad \text{with} \quad 
\widetilde{H}^{(0)} = D \;\;\text{and}\;\; \widetilde{H}^{(k)} = \mathcal R (S^{(k-1)} D + Q^{(k-1)}V) \;\,\text{for} \;\, k\ge 1.
\end{equation*}

\item
For each $k \ge 0$, $\widetilde{H}^{(k)}$ is a function belonging to $\mathcal S(\Omega)$, and $R^{(k)}$ an operator on $\mathcal S(\Omega)$.
Let $f=\sum_{x\in\Z_N} f_x \in \mathcal S(\Omega)$ be given.
The functions $\widetilde H^{(k)}$ and $R^{(k)} f$ can be decomposed as a sum of local terms, with for example, for $k\ge 1$, 
\begin{equation*}
\widetilde{H}^{(k)}_x \; = \; \mathcal R (S^{(k-1)} D_x + Q^{(k-1)}V_x)
\qquad \text{and} \qquad 
(R^{(k)}f)_x \; = \; R^{(k)} f_x.
\end{equation*}
Moreover, we will show in Subsection \ref{sub: proofs change of variable} below, that there exists an integer $r_k$ such that 
\begin{equation}\label{dependence on r in change of variables}
r \big(\widetilde{H}^{(k)} \big) \; \le \; r_k
\qquad \text{and} \qquad
r \big( R^{(k)} f \big) \; \le \; r_k + r(f),
\end{equation}
where $r$ is the parameter introduced in \eqref{parameter r for functions in S}.

\item
The function $\widetilde{H}^{(k)}$ and the operator $R^{(k)}$ depend on $\delta$ if $k\ge 1$. 
Let thus $k\ge 1$.
In what follows, we will use the symbol $b$ to denote smooth, bounded functions on $\R^N \times (0,1)$ with bounded derivatives of all order, and 
the symbol $f$ for functions in $\mathcal S (\Omega)$. 
We will show the two following assertions in Subsection \ref{sub: proofs change of variable} below. 
First, there is an integer $m_k$ such that, given $x\in \Z_N$,  $\widetilde H_x^{(k)}$ can be expressed as a sum of the type
\begin{equation}\label{H depends on delta}
\widetilde{H}_x^{(k)} (q,\omega ;\delta)
\; = \;
\delta^{-2 (k-1)} 
\sum_{j=1}^{m_k} b_{j,x} (\omega/\delta, \delta) \, f_{j,x} (q,\omega)
\end{equation}
such that the functions $b_{j,x}$ and $f_{j,x}$ depend on the same variables and the same Fourier modes as $\widetilde{H}_x^{(k)}$ and the bounds on them can be chosen uniformly in $x$.
Second, consider a function $g(\,\cdot \, ; \delta) \in \mathcal S (\Omega)$ such that $g_x(q, \omega ; \delta) = b_x(\omega/\delta,\delta) f_x (q, \omega)$.
Then, there is an integer $m_{k,g}$ such that $(R^{(k)} g)_x$ can be expressed as a sum of the type
\begin{equation}\label{R depends on delta}
(R^{(k)} g)_x (q,\omega ; \delta)
\; = \; 
\delta^{-2k}\sum_{j=1}^{m_{k,g}} b'_{j,x} (\omega/\delta,\delta) f'_{j,x} (q, \omega)
\end{equation}
and such that the functions $b'_{j,x}$ and $f'_{j,x}$ depend on the same variables and the same Fourier modes as $(R^{(k)} g)_x$.
\end{enumerate}

\subsection{Proof of Proposition \ref{pro: change of variable} and relations 
(\ref{dependence on r in change of variables}-\ref{R depends on delta})}\label{sub: proofs change of variable}

\begin{proof}[Proof of Proposition \ref{pro: change of variable}]
Given a function $U\in \mathcal S(\Omega)$, a formal change of variable, seen as an operator on $\mathcal S (\Omega)$, is defined through
\begin{equation*}
\ed^{\epsilon L_U} \; = \; \sum_{k\ge 0} \frac{\epsilon^k}{k!} L_U^k.
\end{equation*}
Given now a sequence $(U^{(k)})_{k\ge 1}\subset \mathcal S(\Omega)$, that we later will identify with the sequence defined by \eqref{definition U k}, 
we construct the formal change of variable
\begin{align*}
Q
\; &= \;
\dots \ed^{\epsilon^n L_{U^{(n)}}} \dots  \ed^{\epsilon^2 L_{U^{(2)}}} \ed^{\epsilon L_{U^{(1)}}} 
\; = \; 
\sum_{j_1 \ge 0, \dots , j_n\ge 0 , \dots } \frac{\epsilon^{j_1 + \dots + nj_n + \dots }}{j_1 ! \dots j_n! \dots } \big(  \dots L_{U^{(n)}}^{j_n}\dots L_{U^{(1)}}^{j_1} \big) \\
\; &= \; 
\Id + \sum_{k\ge 1} \epsilon^k \sum_{\underj \in \pi(k)} \frac{1}{j_1! \dots j_k!} L_{U^{(k)}}^{j_k } \dots L_{U^{(1)}}^{j_1}
\; = \; 
\sum_{k\ge 0} \epsilon^k Q^{(k)}. 
\end{align*}
The second line was obtained using the definition of $\underj \in \pi(k)$ introduced at the beginning of Subsection \ref{subsection: statement of result in perturbative part}.
The formal inverse of $Q$ is given by 
\begin{align*}
R
\; &= \;
\ed^{-\epsilon L_{U^{(1)}}} \ed^{-\epsilon^2 L_{U^{(2)}}} \dots \ed^{-\epsilon^n L_{U^{(n)}}} \dots \\
\; &= \; 
\Id + \sum_{k\ge 1} \epsilon^k \sum_{\underj \in \pi(k)} \frac{(-1)^{j_1 + \dots + j_n}}{j_1! \dots j_k!}
 L_{U^{(1)}}^{j_1} \dots L_{U^{(k)}}^{j_k} 
\; = \; 
\sum_{k\ge 0} \epsilon^k R^{(k)}. 
\end{align*}

Let us show the first part of Proposition \ref{pro: change of variable}.
The operators $Q$ and $R$ are formal inverse of each others, so that,   
for every $f\in \mathcal S (\Omega)$ such that $\mathcal T_{n_1} f = f$ (see the remark after \eqref{troncature}), it holds that
\begin{equation*}
f \; = \; \mathcal T_{n_1} \big( R \mathcal T_{n_1} (Qf)\big),
\end{equation*}
as can be checked by a direct computations with formal series. 
We will thus be done if we show that 
\begin{equation}\label{Hn tilde = Tn (QH)}
\widetilde{H}_{n_1} \; = \; \mathcal T_{n_1} (QH)
\end{equation}
We compute
\begin{equation*}
Q H 
\; = \;
\sum_{k\ge 0} \epsilon^k Q^{(k)} (D + \epsilon V)
\; = \; 
D + \sum_{k\ge 1} \epsilon^k (Q^{(k)} D + Q^{(k-1)}V). 
\end{equation*}
It holds that 
\begin{equation*}
Q^{(k)} \; = \; S^{(k-1)} + L_{U^{(k)}} \quad \text{for} \quad k\ge 1. 
\end{equation*}
Since $L_{U^{(k)}} D = - L_D U^{(k)}$ for every $k \ge 1$, 
and taking now $U^{(k)}$ as defined by \eqref{definition U k},
we obtain
\begin{equation*}
QH \; = \; D + \sum_{k\ge 1} \epsilon^k (S^{(k-1)} D + Q^{(k-1)}V - L_D U^{(k)})
\; = \; D  + \sum_{k\ge 1} \epsilon^k \mathcal R (S^{(k-1)} D + Q^{(k-1)}V).
\end{equation*}
From this, we derive \eqref{Hn tilde = Tn (QH)}.

Let us then show the second part of Proposition \ref{pro: change of variable}.
The operators $Q$ and $R$ are formal canonical transformations, inverse of each other.
Therefore
\begin{equation}\label{2e brol}
L_H R \; = \; R L_{QH} ,
\end{equation}
as a direct, but lengthy, computation with formal series can confirm. 
Let us next take $f$ such that $f = \mathcal T_{n_1} (f)$. 
By \eqref{Hn tilde = Tn (QH)}, we find that 
\begin{equation*}
\mathcal T_{n_1} \big( R L_{\widetilde{H}_{n_1}} f \big)
\; = \; 
\mathcal T_{n_1} \big( R L_{\mathcal T_{n_1} QH} f \big)
\; = \; 
\mathcal T_{n_1} \big( R  L_{ QH} f \big),
\end{equation*} 
since higher order terms do not contributre thanks to the overall truncation $\mathcal T_{n_1}$. 
Therefore, by \eqref{2e brol},
\begin{align*}
L_H \big( \mathcal T_{n_1} (R f) \big)  
-
\mathcal T_{n_1} \big( R L_{\widetilde{H}_{n_1}} f \big) 
\; &= \; 
L_H \big( \mathcal T_{n_1} (R f) \big)  
-
\mathcal T_{n_1} \big( R  L_{ QH} f \big)\\
\; &= \; 
L_H \big( \mathcal T_{n_1} (R f) \big)   - \mathcal T_{n_1} \big( L_H R f \big).
\end{align*}
Since $L_H = L_D + \epsilon L_V$, it is finally computed that
\begin{equation*}
L_H \big( \mathcal T_{n_1} (R f) \big)   - \mathcal T_{n_1} \big( L_H R f \big)
\; = \; 
\epsilon^{n_1+1} L_V
\sum_{k=0}^{n_1}
R^{{n_1}-k} f^k.
\end{equation*}

Let us finally establish the last part of Proposition \ref{pro: change of variable}.
A function will be said symmetric or antisymmetric if it is symmetric or antisymmetric with respect to the operation $\omega \mapsto - \omega$. 
We observe that, if a function $U$ is symmetric, then $L_U$ exchanges symmetric and antisymmetric functions, 
while $L_U$ preserves the symmetry if $U$ is antisymmetric. 
The action of $\mathcal R$ also preserves the symmetry. 
We deduce that the operation $L_D^{-1}(\Id - \mathcal R)$ exchange symmetric and antisymmetric functions, since $D$ is symmetric. 
It is then recursively established from (\ref{definition U k}-\ref{definition S k}) that the functions $U_k$ are antisymmetric for $k\ge 0$, 
while the operators $Q_k$, $R_k$ and $S_k$ preserve the symmetry for $k\ge 0$. 
Since $D$ and $V$ are symmetric, we conclude from \eqref{definition of H tilde} that $\widetilde H$ is symmetric.
\end{proof}

\begin{proof}[Proof of (\ref{dependence on r in change of variables}-\ref{R depends on delta})]
Let us first establish \eqref{dependence on r in change of variables}.
Given two functions $f,g\in \mathcal S(\Omega)$, 
the function $L_g f$ is decomposed as a sum of local terms $(L_g f)_x$, that we have chosen to be given by $(L_g f)_x = L_g f_x$. 
A direct computations shows that $r(L_g f) \le 2 r(g) + r (f)$. 
Since $r(L_D^{-1}(\Id - \mathcal R)f) = r(f)$, we readily deduce \eqref{dependence on r in change of variables} from (\ref{definition U k}-\ref{definition S k}).

Let us next show \eqref{H depends on delta} and \eqref{R depends on delta}.
Since we are only interested in tracking the dependence on $\delta$, we may simplify notations as much as possible in the following way. 
We use the symbols $b$ and $f$ with the same meaning as in the paragraph where \eqref{H depends on delta} and \eqref{R depends on delta} are stated.
Let $n\ge 0$.
First, if $g\in \mathcal S(\Omega)$, we just write $g \sim \delta^{-n}$ to express that $g$ is of the following form: 
$g=\sum_x g_x$  as in \eqref{f sum of local terms} and $g_x$ take the form $g_x (q,\omega;\delta) = \delta^{-n}\sum_j b_{j,x} (\omega/\delta,\delta)f_{j,x}(\omega,q)$ 
with all bounds on $b_{j,x}, f_{j,x}$ uniform in $x$.
Next, if $A$ is an operator on $\mathcal S(\Omega)$, we just write $A \sim \delta^{-n}$
to express that, for any $h\in \mathcal S(\Omega)$  such that $h \sim \delta^{-m}$,  we have $A h \sim \delta^{-n-m}$. 

We now observe that, if $g\sim \delta^{-n}$ and $h \sim \delta^{-m}$, then $L_g h \sim \delta^{-(n+m+1)}$, and that
if $u$ solves the equation $L_D u =(\Id - \mathcal R)g$ and if $g\sim \delta^{-n}$, then $u \sim \delta^{-(n+1)}$.
It is then established recursively that, for $k\ge 1$, we have
\begin{equation}\label{to be shown for the delta dependence}
Q^{(k-1)} \;  \sim \; \delta^{-2(k-1)}, 
\;\;
R^{(k-1)} \; \sim \; \delta^{-2(k-1)}, 
\;\;
S^{(k-1)}D \; \sim \; \delta^{-2(k-1)}, 
\;\; 
U^{(k)} \; \sim \; \delta^{-(2k-1)}, 
\end{equation}
from which \eqref{H depends on delta} and \eqref{R depends on delta} are readily derived. 
By the definitions, the relations \eqref{to be shown for the delta dependence} hold for $k = 1$. 
Let us see that the claim for $1, \dots, k\ge 1$ implies the claim for $k+1$. 

Let us start with $Q^{(k)}$. 
For $\underj \in \pi(k)$, since $j_1 + 2j_2 + \dots + kj_k = k$, we get from the definition \eqref{definition Q k} that
\begin{equation*}
Q^{(k)} 
\; \sim \;
(\delta^{-(2k-1) -1})^{j_k } \dots (\delta^{-(2-1)-1})^{j_1} \; = \; \delta^{-2k}.
\end{equation*}
The case of $R^{(k)}$ is handled in the same way.  
Let us then treat $S^{(k)}D$. 
We decompose $S^{(k)}D = \sum_{\underj} S^{(k)}_{\underj} D$ according to the definition \eqref{definition S k},
we pick one of the sequences $\underj$, and we let $l\ge 1$ be the smallest integer such that $j_l \ge 1$.
Thanks to \eqref{definition U k} and to our inductive hypothesis, and because of the constraint $ j_1+ 2 j_2 +  \dots + k j_k = k+1$ on $\underj \in \pi(k+1)$,  
we get, for some constant $\mathrm C(\underj)$,
\begin{align*}
S^{(k)}_{\underj} D 
\; &= \; 
\mathrm C(\underj)\,
L_{U^{(k)}}^{j_k} \dots  L_{U^{(l)}}^{j_l- 1} \big( L_{U^{(l)}} D  \big)\\
\; &= \; 
- \mathrm C (\underj)\,
 L_{U^{(k)}}^{j_k} \dots  L_{U^{(l)}}^{j_l - 1} \big( (\mathrm{Id} - \mathcal R) (S^{(l-1)}) D + Q^{(l-1)} V \big) \\
\; &\sim \; 
\delta^{-2 \{ l(j_l - 1) + \dots + k j_k  \}} \delta^{-2(l-1)}
\; = \; 
\delta^{-2(k+1) + 2l - 2(l-1)}
\; = \; 
\delta^{-2k}. 
\end{align*}
So we conclude that $S^{(k)} D \sim \delta^{-2k}$.
The statement for $U^{(k+1)}$ is finally derived using \eqref{definition U k}. 
\end{proof}

\section{Resonant frequencies}\label{sec: resonant frequencies}

Given a point $x\in \Z_N$, we construct a subset $\mathsf R(x)$ of the frequencies $\omega$, 
seen as a subset of the full phase space $\Omega$ that does not depend on the positions $q$, with the two following characteristics.  
First, if a state does not belong to this set, then the energy current for the Hamiltonian $\widetilde H$ vanishes through the bonds near $x$.
Second, it is approximately invariant under the dynamics generated by $\widetilde{H}$, 
meaning that in a small time interval, only the  frequencies in a subset $\mathsf S(x)$, of small probability with respect to the Gibbs measure, can leave or enter the set $\mathsf R (x)$. 

In our opinion, the ideas of this Section are best understood visually. 
We hope that figure \ref{figure: facebook group} will help in that respect (see below for the definition of the set $\mathsf B(\modek_1,\modek_2)$).
We let 
\begin{equation}\label{definition du parametre r}
r \; = \; r(n_1) \; = \; \max_{1 \le k \le n_1} r_k, 
\end{equation}
where the numbers $r_k$ are defined in \eqref{dependence on r in change of variables}. 
We let $\delta > 0$ be as in Section \ref{sec: approximate change of variables}.

\subsection{Preliminary definitions}

We recall that, given $\modek\in \Z^{N}$, we denote by $\supp (\modek) \subset \Z^{N}$ the set of points $x$ such that $\modek_x \ne 0$.
We define the set $K_r \subset \Z^{N}$ of vectors $\modek = (\modek_x)_{x\in\Z_N}$ such that 
$\max_{x\in\Z_N} |\modek_x| \le r$ and $\supp (\modek) \subset \mathrm B (r)$ for some ball $\mathrm B(r)$ of radius $r$.  
We write $\str \modek \str^2_2=\sum_x\str \modek_x \str^2$.
One easily checks that for any $\modek \in K_r$ and $r>1$, we have $\str \modek \str_2\leq r^2$ and this will be used without further comment.

Given $x\in\Z^d$, we say that a subset $\{ \modek_1, \dots , \modek_p \} \subset K_r$ is a cluster around $x$ if 
\begin{enumerate}
\item 
the vectors $\modek_1, \dots , \modek_p$ are linearly independent,
\item 
if $p\ge 2$, for all $1 \le i \ne j \le p$, there exist $1 \le i_1, \dots , i_m \le p$ 
such that $i_1=i$, $i_m=j$ and $\supp (\modek_{i_s}) \cap \supp (\modek_{i_{s+1}}) \ne \varnothing$ for all $1 \le s \le m-1$,


\item 
$\supp (\modek_j)\subset \mathrm B(x,4r)$ for some $1 \le j \le p$.  
\end{enumerate}

Finally, given $\modek\in K_r$, we define
\begin{equation*}
\pi (\modek) \; = \; \{ \omega \in \R^{N} : \modek\cdot \omega = 0 \}.
\end{equation*}
Given a subspace $E \subset \R^{N}$, and given $\omega \in \R^{N}$, we denote by $P(\omega, E)$ the orthogonal projection of $\omega$ on the subspace $E$.

\subsection{Approximately invariant sets of resonant frequencies}

Let $L > 0$, 
let $n_2\ge 1$, 
and let $x\in \Z_N$. 
Let us define two subsets of $\R^{N}$: 
a set $\mathsf R_{\delta,n_2}(x) \subset\R^{N}$ of resonant frequencies, 
and a small set $\mathsf S_{\delta,n_2}(x) \subset\R^{N}$ of ``multi-resonant" frequencies.

To define $\mathsf R_{\delta,n_2}(x)$, let us first define the sets $\mathsf B_\delta (\modek_1, \dots , \modek_p) \subset \R^{N}$, 
where $\{ \modek_1, \dots , \modek_p \}$ is a cluster around $x$.
We say that $\omega \in \mathsf B_\delta (\modek_1, \dots , \modek_p)$ if
\begin{equation}\label{first condition B set}
\big| \omega - P \big( \omega , \pi (\modek_1) \cap \dots \cap \pi (\modek_p) \big) \big|_2
\; \le \; L^p \delta
\end{equation}
and if, for every linearly independent $\modek'_1, \dots , \modek'_{p'} \in K_r \cap \spanv \{ \modek_1, \dots ,\modek_p\}$, 
\begin{equation*}
\big| P \big( \omega , \pi (\modek_1') \cap \dots \cap \pi (\modek'_{p'}) \big) - P \big( \omega , \pi (\modek_1) \cap \dots \cap \pi (\modek_p) \big) \big|_2
\; \le \;
\big( L^p - L^{p'} \big) \delta .
\end{equation*}
We next define $\mathsf R_{\delta,n_2}(x)$ as the union of all the sets $\mathsf B_\delta (\modek_1, \dots , \modek_p) \subset \R^{N}$ with $p \le n_2$.

We then define $\mathsf S_{\delta,n_2} (x)$ as the set of points $\omega\in \R^{N}$ for which there exists a cluster $\{ \modek_1 , \dots , \modek_{n_2} \}$ around $x$, 
such that $|\modek_j \cdot \omega| \le L^{n_2+1} \delta$ for every $1 \le j \le n_2$.

We finally define a smooth indicator function of the complement of $\mathsf R_{\delta,n_2}(x)$ by means of a convolution: 
\begin{equation}\label{smooth indicator good set}
\theta_{x,\delta,n_2} (\omega) 
\; = \;
1 - 
\frac{1}{\Big( \int_\R \rho_\delta (z) \, \dd z \Big)^{N}} \;
\int_{\R^{N}} \chi_{\mathsf R_{\delta,n_2} (x)} (\omega + \omega') \Big( \prod_{x\in \Z_N} \rho_\delta (\omega'_x) \Big) \, \dd \omega'. 
\end{equation}
This naturally may be seen as a function on the full phase space $\Omega$ that is independent of the $q-$variable.

\begin{Proposition}\label{pro: invariance resonant set}
Let $n_1$ be given, and so $r(n_1)$ defined by \eqref{definition du parametre r} be fixed as well. 
Let then $n_2 \ge 1$ be fixed. 
The following holds for $L$ large enough.
\begin{enumerate}
\item
If
$\theta_{x,\delta,n_2} (\omega) > 0$
then 
$\rho_\delta (\omega \cdot \modek) = 0$
for all
$\modek\in K_r$
such that 
$\supp (\modek) \subset \mathrm B(x,4r)$.
\item
$L_{\widetilde H_{n_1}} \theta_{x,\delta,n_2} (q,\omega)= 0$
for all 
$(q,\omega) \in\Omega$
such that 
$q\in \T^{N}$ 
and
$\omega \notin \mathsf S_{n_2} (x)$.
\end{enumerate}
\end{Proposition}

\begin{figure}[t]
\begin{center}
\begin{tikzpicture}[scale=0.5]

\draw (-7,0) -- (8.5,0) ;
\draw (245:6) -- (65:8) ;

\draw [very thick] (-25:5) -- (25:5) ;
\draw [very thick] (155:5) -- (205:5) ;
\draw [very thick] (40:5) -- (90:5) ;
\draw [very thick] (220:5) -- (270:5) ;

\draw [color=gray!60] (0,0) circle (5) ;

\draw [very thick] (25:5) arc (25:40:5) ;
\draw [very thick] (90:5) arc (90:155:5) ;
\draw [very thick] (205:5) arc (205:220:5) ;
\draw [very thick] (270:5) arc (270:335:5) ;

\draw [thick,>=stealth,->] (5.8,0) -- (5.8,3) ;
\draw [rotate=65,thick,>=stealth,->] (5.8,0) -- (5.8,3) ;

\draw (5.9,1.6) node[right]{{\Large$\modek_1$}} ;
\draw (1.2,6) node[above]{{\Large$\modek_2$}} ;

\draw (7.2,0) node[below]{{\Large$\pi (\modek_1)$}} ;
\draw (2.9,6.2) node[right]{{\Large$\pi (\modek_2$)}} ;

\draw[>=stealth,<->] (-50:0.03) -- (-50:4.97) ; 

\draw (0,0) -- (155:1) ;
\draw[xshift=-15,yshift=7,>=stealth,<->] (0:0) -- (65:4.5) ; 

\draw(2.4,-1.5) node(above){{\large$L^2\delta$}} ;
\draw(-1.5,2.3) node(left){{\large$(L^2-L)\delta$}} ;
\end{tikzpicture}
\end{center}

\caption{\label{figure: facebook group}
The set $\mathsf B_\delta(\modek_1,\modek_2)$.
The plane is the subspace of points of the form $\omega - P(\omega, \pi (\modek_1) \cap \pi (\modek_2))$, for $\omega \in \R^N$.
We have drawn a disk of radius $L^2 \delta$ that is `flattened' by an amount $L\delta$  at the intersection of its boundary with the lines $\pi(\modek_1),\pi(\modek_2)$. 
To simplify the figure, we have pretended that $\modek_1$ and $\modek_2$ are the only vectors in $\spanv \{\modek_1,\modek_2\}\cap K_r$.
This is not so in reality: the disk still needs to be flattened by an amount $L\delta$ 
at each intersection point of its boundary with a line $\pi(\modek)$ for all $\modek \in \spanv \{\modek_1,\modek_2\}\cap K_r$. 
The set $\mathsf B(\modek_1,\modek_2,\modek_3)$ could be similarly visualized as a ball of radius $L^3 \delta$, 
that is flattened by $L\delta$ along the circles corresponding to the intersection of its boundary with a plane $\pi (\modek)$, 
and flattened by $L^2 \delta$ at the points where its boundary intersect a line $\pi (\modek) \cap \pi (\modek')$, 
for all $\modek,\modek' \in \spanv \{\modek_1,\modek_2,\modek_3\}\cap K_r$.
}
\end{figure}
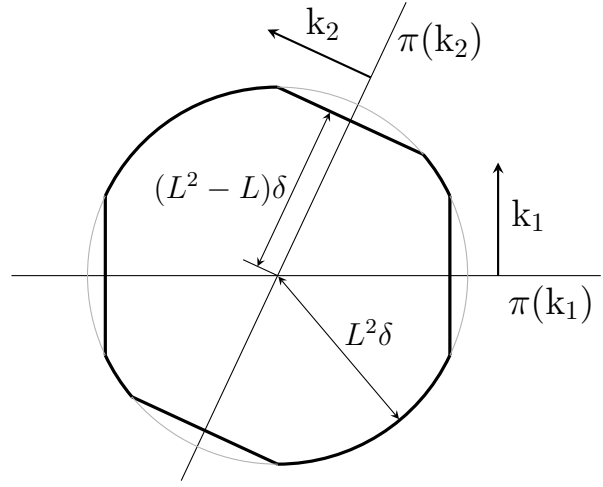

\subsection{Proof of Proposition \ref{pro: invariance resonant set}}\label{sub: proof of the invariance of the resonant set}

We start by a series of lemmas. 
The first one simply expresses, in a particular case, that if a point is close to two vector spaces, then it is also close to their intersection. 
The uniformity of the constant $\mathrm C$ comes from the fact that we impose the vectors to sit in the set $K_r$.

\begin{Lemma}\label{lem: close to two subspaces implies close to the intersection}
Let $p\ge 1$.
There exists a constant $\mathrm C = \mathrm C (r,p) < + \infty$ such that, 
given linearly independent vectors $\modek_1, \dots , \modek_p,\modek_{p+1} \in K_r$ and given $\omega \in \R^{N}$, it holds that 
\begin{multline*}
\big| \omega - P\big( \omega, \pi (\modek_1) \cap \dots \cap \pi (\modek_{p}) \cap \pi (\modek_{p+1}) \big) \big|_2
\\ \le \; \mathrm C \,
\Big( 
|\modek_{p+1} \cdot \omega | + \big| \omega - P\big( \omega, \pi (\modek_1) \cap \dots \cap \pi (\modek_{p}) \big) \big|_2
\Big) .
\end{multline*}
\end{Lemma}

\begin{proof}
First, 
\begin{multline}\label{first inequality in close to two subspaces implies close to the intersection}
\big| \omega - P\big( \omega, \pi (\modek_1) \cap \dots \cap  \pi (\modek_{p+1}) \big) \big|_2
\; \le \; 
\big| \omega - P\big( \omega, \pi (\modek_1) \cap \dots \cap \pi (\modek_{p})  \big) \big|_2
\\
+ 
\big| P\big( \omega, \pi (\modek_1) \cap \dots \cap  \pi (\modek_{p+1}) \big) - P\big( \omega, \pi (\modek_1) \cap \dots \cap \pi (\modek_{p})  \big) \big|_2.
\end{multline}
The lemma is already shown if the second term in the right hand side is zero. We further assume this  not to be the case. 
Next, since $\modek_{p+1} \cdot P\big( \omega, \pi (\modek_1) \cap \dots \cap \pi (\modek_{p+1}) \big) = 0$, we obtain
\begin{multline*}
\modek_{p+1} \cdot \omega 
\; = \; 
\modek_{p+1} \cdot \Big( \omega - P\big( \omega, \pi (\modek_1) \cap \dots \cap \pi (\modek_{p})  \big) \Big)
\; + \\ 
\modek_{p+1} \cdot \Big( 
P\big( \omega, \pi (\modek_1) \cap \dots \cap \pi (\modek_{p}) \big) - P\big( \omega, \pi (\modek_1) \cap \dots \cap \pi (\modek_{p+1}) \big)
\Big).
\end{multline*}
This implies
\begin{multline}\label{second inequality in close to two subspaces implies close to the intersection}
\Big| \modek_{p+1} \cdot \Big( 
P\big( \omega, \pi (\modek_1) \cap \dots \cap \pi (\modek_{p}) \big) - P\big( \omega, \pi (\modek_1) \cap \dots \cap \pi (\modek_{p+1}) \big) \Big)
\Big|
\; \le 
\\
|\modek_{p+1} \cdot \omega |
+ 
|\modek_{p+1}|_2
\Big|
\omega - P\big( \omega, \pi (\modek_1) \cap \dots \cap \pi (\modek_{p})  \big) 
\Big|_2.
\end{multline}
The vector
\begin{equation}\label{definition of v in close to two subspaces implies close to the intersection}
v \; = \; \frac{P\big( \omega, \pi (\modek_1) \cap \dots \cap \pi (\modek_{p}) \big) - P\big( \omega, \pi (\modek_1) \cap \dots \cap \pi (\modek_{p+1}) \big) }
{\big| P\big( \omega, \pi (\modek_1) \cap \dots \cap \pi (\modek_{p}) \big) - P\big( \omega, \pi (\modek_1) \cap \dots \cap \pi (\modek_{p+1}) \big) \big|_2}
\end{equation}
is well defined since we have assumed that the denominator in this expression does not vanish. 
The bound \eqref{second inequality in close to two subspaces implies close to the intersection} is rewritten as
\begin{multline*}
\big| P\big( \omega, \pi (\modek_1) \cap \dots \cap \pi (\modek_{p}) \big) - P\big( \omega, \pi (\modek_1) \cap \dots \cap \pi (\modek_{p+1}) \big) \big|_2
\; \le 
\\
\frac{|\modek_{p+1} \cdot \omega |
+ 
|\modek_{p+1}|_2
\Big|
\omega - P\big( \omega, \pi (\modek_1) \cap \dots \cap \pi (\modek_{p})  \big) 
\Big|_2}{|\modek_{p+1} \cdot v|}
\end{multline*}
Inserting this last inequality in \eqref{first inequality in close to two subspaces implies close to the intersection}, we arrive at 
\begin{multline*}
\big| \omega - P\big( \omega, \pi (\modek_1) \cap \dots \cap  \pi (\modek_{p+1}) \big) \big|_2
\; \le \\ 
\frac{|\modek_{p+1} \cdot \omega |}{|\modek_{p+1} \cdot v|}
+ 
\bigg( 1 + \frac{|\modek_{p+1}|_2}{|\modek_{p+1} \cdot v|} \bigg) 
\Big|
\omega - P\big( \omega, \pi (\modek_1) \cap \dots \cap \pi (\modek_{p})  \big) 
\Big|_2 .
\end{multline*}

To finish the proof, it remains to establish that $|\modek_{p+1} \cdot v|$ can be bounded from below by some strictly positive constant, 
where $v$ is given by \eqref{definition of v in close to two subspaces implies close to the intersection}. 
Let us show that
\begin{equation}\label{v is one dimensional}
v = \pm \frac{P \big(\modek_{p+1}, \pi (\modek_1) \cap \dots \cap \pi (\modek_p)\big)}{\big| P \big(\modek_{p+1}, \pi (\modek_1) \cap \dots \cap \pi (\modek_p)\big) \big|_2}.
\end{equation}
We can find vectors $\modek_{p+2}, \dots , \modek_N$ so that $\{ \modek_1, \dots, \modek_N \}$ forms a basis of $\R^N$
and so that every vector $k_j$ with $p+2 \le j \le N$ is orthrogonal to $\spanv \{ \modek_1, \dots , \modek_{p+1} \}$. 
We express the vector $\omega$ in this basis, $\omega = \sum_{j=1}^N \omega^j \modek_j$, and, from \eqref{definition of v in close to two subspaces implies close to the intersection}, 
we deduce that, for some non-zero constant $R$, we have
\begin{equation*}
v \; = \; R\sum_{j=1}^N \omega^j \Big\{ P\big( \modek_j, \pi (\modek_1) \cap \dots \cap \pi (\modek_{p}) \big) - P\big( \modek_j, \pi (\modek_1) \cap \dots \cap \pi (\modek_{p+1}) \big)  \Big\}. 
\end{equation*}
All the terms corresponding to $1 \le j \le p$ vanish since $\modek_j \perp \pi (\modek_j)$, 
the term $P\big( \modek_{p+1}, \pi (\modek_1) \cap \dots \cap \pi (\modek_{p+1}) \big)$ vanishes for the same reason, 
and all the terms corresponding to $j \ge p+2$ vanish too, as they read in fact $\modek_j - \modek_j = 0$. 
So, the only term left is $R \omega^{p+1}P\big( \modek_{p+1}, \pi (\modek_1) \cap \dots \cap \pi (\modek_{p}) \big)$ and, since $|v|=1$, we arrive at \eqref{v is one dimensional}. 

From \eqref{v is one dimensional} we deduce that 
\begin{equation*}
|v\cdot \modek_{p+1}| \; = \;
\big| P \big(\modek_{p+1}, \pi (\modek_1) \cap \dots \cap \pi (\modek_p)\big) \big|_2 .
\end{equation*}
If $\modek_{p+1} \bot \spanv \{ \modek_1, \dots , \modek_p \}$, then the right hand side just becomes $|\modek_{p+1}|_2$.
This quantity is bounded from below by a strictly positive constant since so is the the norm of any nonzero vector in $K_r$. 
Otherwise, if $\modek_{p+1} \cancel{\bot} \spanv \{ \modek_1, \dots ,\modek_p \}$, 
we know however that the quantity cannot vanish since $\modek_{p+1} \notin \spanv\{ \modek_1, \dots , \modek_p \}$.
Because they are only finitely many vectors $\modek \in K_r$ with the property that $\modek \cancel{\bot} \spanv \{ \modek_1, \dots , \modek_p \}$, we conclude that 
the quantity is bounded from below by a strictly positive constant.
\end{proof}

The next Lemma describes the crucial geometrical properties of the sets $\mathsf B_\delta (\modek_1, \dots ,\modek_p)$ 
that allows to establish the second assertion of Proposition \ref{pro: invariance resonant set}.

\begin{Lemma}\label{lem: invariance of the sets B}
Let $\{ \modek_1 , \dots , \modek_p \}$ be a cluster around $x$. 
If, given $\mathrm K < + \infty$, $L$ is taken large enough, then, 
for every $\modek \in K_r \cap \spanv \{\modek_1, \dots ,\modek_p\}$, it holds that if
\begin{equation*}
\omega \in \mathsf B_\delta (\modek_1 , \dots , \modek_p) \quad \text{and} \quad  |\modek \cdot \omega| \; \le \; \mathrm K \delta
\end{equation*}
then
\begin{equation*}
\omega + t \modek \in \mathsf B_\delta (\modek_1 , \dots , \modek_p)
\quad \text{as long as} \quad  |t| \le \delta.  
\end{equation*}
\end{Lemma}

\begin{proof}
To simplify some further expressions, let us define
\begin{equation*}
\omega' = \omega - P (\omega, \pi (\modek_1) \cap \dots \cap \pi(\modek_p)) \in \spanv(\modek_1, \dots , \modek_p).
\end{equation*}
The conditions ensuring that $\omega \in \mathsf B(\modek_1, \dots , \modek_p)$ now simply read
\begin{equation}\label{new formulation special set conditions}
|\omega'|_2 \; \le \; L^p \delta \quad \text{and} \quad
\big|P\big(\omega', \pi (\modek_1') \cap \dots\cap \pi (\modek'_{p'}) \big)\big|_2 \; \le \; \big( L^p - L^{p'} \big) \delta \quad (p' < p),
\end{equation}
for all linearly independent $\modek_1', \dots , \modek_{p'}' \in \modek \in K_r \cap \spanv \{\modek_1, \dots ,\modek_p\}$. 
The condition $|\modek \cdot \omega| \le \mathrm K \delta$ implies $|\modek \cdot \omega'|  \le  \mathrm K \delta$.
We need to show that
\begin{align}
&|\omega' + t \modek |_2 \; \le \; L^p \delta&
 \text{for} \qquad |t|\le \delta,
\label{point 1 in proof invariance set B}\\
&\big|P\big(\omega' + t\modek, \pi (\modek_1') \cap \dots\cap \pi (\modek'_{p'}) \big)\big|_2 \; \le \; \big( L^p - L^{p'} \big) \delta&
 \text{for} \qquad |t|\le \delta.
\label{point 2 in proof invariance set B}
\end{align}

Let us start with \eqref{point 1 in proof invariance set B}:
\begin{align*}
|\omega' + t \modek|_2
\; &\le \;
|\omega' - P(\omega',\pi (\modek))|_2 + |P(\omega',\pi (\modek))|_2 + |t| |k|_2\\
\; &\le \;
|\omega' \cdot \modek| + |P(\omega',\pi (\modek))|_2 + |t| |k|_2 \\
\; &\le \; 
\mathrm K \delta + (L^p - L) \delta + r^2 \delta
\; \le \; L^p \delta.
\end{align*}
Here, to get the penultimate inequality, we have used \eqref{new formulation special set conditions} and the hypothesis $\modek \in K_r \cap \spanv \{\modek_1, \dots ,\modek_p\}$,
implying in particular $|\modek|_2 \le r^2$, 
while the last inequality is valid for large enough $L$.

Let us next move to \eqref{point 2 in proof invariance set B}. 
Let us fix $\modek_1', \dots , \modek'_{p'}$. 
It is seen that, if $\modek \in \spanv \{ \modek_1', \dots , \modek'_{p'} \}$, then  \eqref{point 2 in proof invariance set B} is actually satisfied for all $t \in \R$. 
Let us therefore assume $\modek\notin \spanv \{ \modek_1', \dots , \modek'_{p'} \}$.
We write also $\modek = \modek_{p'+1}$. 
We will show that, because $ |\modek \cdot \omega| \; \le \; \mathrm K \delta$, then in fact 
\begin{equation}\label{thing to show in invariance of the sets B}
\big|P\big(\omega', \pi (\modek_1') \cap \dots\cap \pi (\modek'_{p'}) \big)\big|_2 \; \le \; \big( L^p - L^{p'} - r^2 \big) \delta.
\end{equation}
Since $|\modek|_2 \le r^2$, this will imply  \eqref{point 2 in proof invariance set B}.

To establish \eqref{thing to show in invariance of the sets B}, we start by writing the decompositions 
\begin{align}
|\omega'|_2^2 \; = & \; \big| \omega' - P \big(\omega', \pi (\modek'_1) \cap \dots \cap \pi (\modek'_{p'}) \big) \big|_2^2  + \big| P \big(\omega', \pi (\modek'_1) \cap \dots \cap \pi (\modek'_{p'})\big) \big|_2^2, 
\label{decomp 1 in invariance of the sets B}\\
|\omega'|_2^2 \; = & \; \big| \omega' - P \big(\omega', \pi (\modek'_1) \cap \dots \cap \pi (\modek'_{p'+1})\big) \big|_2^2  + \big| P \big(\omega', \pi (\modek'_1) \cap \dots \cap \pi (\modek'_{p'+1})\big) \big|_2^2.
\label{decomp 2 in invariance of the sets B}
\end{align}
We bound the first term in the right hand side of \eqref{decomp 2 in invariance of the sets B} by applying Lemma \ref{lem: close to two subspaces implies close to the intersection}
and then using \eqref{decomp 1 in invariance of the sets B}:
\begin{align*}
&\big| \omega' - P \big(\omega', \pi (\modek'_1) \cap \dots \cap \pi (\modek'_{p+1})\big) \big|_2^2 \\
& \le \; 
\mathrm C \Big( |\modek\cdot \omega'|^2 +  \big| \omega' - P \big(\omega', \pi (\modek'_1) \cap \dots \cap \pi (\modek'_p)\big) \big|_2^2 \Big) \\
&\le \;
\mathrm C \Big( 
|\modek\cdot \omega'|^2 + |\omega'|^2_2 - \big| P \big(\omega', \pi (\modek'_1) \cap \dots \cap \pi (\modek'_p)\big) \big|_2^2
\Big).
\end{align*}
It may be assumed that $\mathrm C \ge 1$.
Reinserting this bound in \eqref{decomp 2 in invariance of the sets B} yields
\begin{align}
|\omega'|_2^2 \; \le &\; 
\mathrm C \Big( 
|\modek\cdot \omega'|^2 + |\omega'|^2_2 - \big| P \big(\omega', \pi (\modek'_1) \cap \dots \cap \pi (\modek'_p)\big) \big|_2^2 \Big) \nonumber\\
&\; + 
\big| P \big(\omega', \pi (\modek'_1) \cap \dots \cap \pi (\modek'_{p'+1})\big) \big|_2^2
\nonumber\\
\le &\; 
\mathrm C \Big( 
\mathrm K^2 \delta^2 + |\omega'|^2_2 - \big| P \big(\omega', \pi (\modek'_1) \cap \dots \cap \pi (\modek'_p)\big) \big|_2^2 \Big)
+ 
(L^p - L^{p' + 1})^2 \delta^2.
\label{inequality that cannot be violated in in invariance of the sets B}
\end{align}
where the hypotheses $|\modek\cdot \omega'| \le \mathrm K \delta$ and $\omega \in \mathsf B(\modek_1, \dots , \modek_p)$ have been used to get the last line. 

Let us now show that \eqref{inequality that cannot be violated in in invariance of the sets B} implies \eqref{thing to show in invariance of the sets B} for $L$ large enough.
For this let us write
\begin{align*}
|\omega'| \; &= \; (1 - \mu)^{1/2} L^p \delta
\qquad \text{with}\qquad 0 \le \mu \le 1, \\
\big| P \big(\omega', \pi (\modek'_1) \cap \dots \cap \pi (\modek'_p)\big) \big|_2 \; &= \; (1 - \nu)^{1/2} (L^p - L^{p'} - r^2) \delta
\qquad \text{with}\qquad \nu \le 1
\end{align*}
($\mu > 0$ actually, thanks to the hypothesis $|\modek\cdot \omega'| \le \mathrm K \delta$, see figure \ref{figure: facebook group}).
Showing \eqref{thing to show in invariance of the sets B} amounts showing $\nu \ge 0$.
With these new notations, inequality \eqref{inequality that cannot be violated in in invariance of the sets B} is rewritten as
\begin{align*}
1 + (\mathrm C - 1) \mu 
\; \le & \;
1 - \frac{2}{L^{p-p'-1}}  \\
&+ \frac{1}{L^{2(p-p'-1)}} + \mathrm C \bigg( 
\frac{\mathrm K^2}{L^{2p}} + \frac{2}{L^{p-p'}} + \frac{2r^2}{L^p} - \frac{1}{L^{2(p-p')}} - \frac{r^4}{L^{2p}} - \frac{2 r^2}{L^{2p-p'}}
\bigg) \\
&+ \mathrm C \nu \bigg( 1 - \frac{L^{p'}+r^2}{L^p} \bigg)^2.
\end{align*}
The left hand side is larger or equal to 1. 
But, when $L$ becomes large, the right hand side is larger or equal to 1 only if $\nu > 0$.
\end{proof}

\begin{Lemma}\label{lem: extension B}
Let $\{ \modek_1 , \dots , \modek_p \}$ be a cluster around $x$, 
and let $\modek \in K_r$ be such that $\modek \notin \spanv\{ \modek_1 , \dots , \modek_p \}$, but such that $\{ \modek_1, \dots , \modek_p,\modek \}$ is a cluster. 
If, given $\mathrm K < + \infty$, $L$ is taken large enough, then
\begin{equation*}
\omega \in  \mathsf B (\modek_1 , \dots , \modek_p) \quad \text{and} \quad  |\modek \cdot \omega| \; \le \; \mathrm K \delta
\qquad  \Rightarrow \qquad 
\omega \in \mathsf B (\modek_1 , \dots , \modek_p,\modek).
\end{equation*}
\end{Lemma}

\begin{proof}
Let us write $\modek = \modek_{p+1}$.
Let $\omega \in  \mathsf B (\modek_1 , \dots , \modek_p)$.
By Lemma \ref{lem: close to two subspaces implies close to the intersection} and by hypothesis, it holds that 
\begin{equation*}
\big| \omega - P \big( \omega, \pi (\modek_1) \cap \dots \cap \pi (\modek_{p+1}) \big) \big|_2
\; \le \;
\mathrm C \, (\mathrm K \delta + L^p \delta) \; \le \; (L-1)L^p \delta
\end{equation*}
if $L$ is large enough. 
Then 
\begin{equation*}
\big| \omega - P \big( \omega, \pi (\modek_1) \cap \dots \cap \pi (\modek_{p+1}) \big) \big|_2
\; \le \;
L^{p+1} \delta
\end{equation*}
and, for every $\modek_1', \dots , \modek'_{p'} \in K_r \cap \spanv (\modek_1, \dots ,\modek_{p+1})$, with $p' < p+1$,
\begin{align*}
&\big| P \big( \omega, \pi (\modek_1') \cap \dots \cap \pi (\modek_{p'}') \big) - P \big( \omega, \pi (\modek_1) \cap \dots \cap \pi (\modek_{p+1}) \big) \big|_2 \\
&\le \;
\big| \omega - P \big( \omega, \pi (\modek_1) \cap \dots \cap \pi (\modek_{p+1}) \big) \big|_2 \\
\; &\le \; \big(L^{p+1} - L^{p}\big) \delta
\; \le\; \big(L^{p+1} - L^{p'}\big) \delta.
\end{align*}
This shows $\omega \in \mathsf B (\modek_1 , \dots ,\modek_p,\modek)$.
\end{proof}

\begin{proof}{Proof of Proposition \ref{pro: invariance resonant set}}
Let us start with the first claim. 
Let $\modek\in K_r$ be such that $\supp (\modek) \subset \mathrm B(x,4r)$, and let $\omega \in \R^N$ be such that $\theta_x (\omega) > 0$. 
On the one hand, from the definition \eqref{smooth indicator good set} of $\theta_x$, it holds that there exists $\omega' \in \R^N$, with $\max_x |\omega_x'| \le 2 \delta$,
such that $\omega + \omega' \notin \mathsf R(x)$. 
On another hand, since $\supp (\modek) \subset \mathrm B(x,4r)$, we conclude that $\{ \modek \}$ alone is a cluster around $x$ so that, 
if $\omega'' \in \R^N$ is such that 
\begin{equation*}
\frac{|\modek \cdot \omega''|}{|\modek|_2} \; = \; \big| \omega'' - P(\omega'',\pi (\modek)) \big| \; \le \; L \delta,
\end{equation*}
then $\omega'' \in \mathsf R(x)$.
We thus conclude that $|(\omega + \omega')\cdot \modek| > |\modek|_2 L \delta \ge L \delta$, and so that 
\begin{equation*}
|\omega \cdot \modek| \; = \; |(\omega + \omega') \cdot \modek \; - \; \omega' \cdot \modek| \; \ge \; L \delta - 2 \delta r^2 \; > \; 2\delta 
\end{equation*}
if $L$ is large enough. We conclude that $\rho_\delta (\omega \cdot \modek) = 0$. 

Let us then show the second part of the Proposition.
Since, by \eqref{definition of H tilde} and \eqref{dependence on r in change of variables}, 
the Hamiltonian $\widetilde{H}$ takes the form
\begin{equation*}
\widetilde{H} (q,\omega)
\; = \; 
\sum_{\modek \in K_r} \rho_\delta (\modek \cdot \omega) {G} (\modek,\omega) \ed^{i \modek\cdot q},
\end{equation*}
for some function $G$ on $\Z^N \times \R^N$, and since the function $\theta_{x}$ is independent of the $q-$variable, it holds that
\begin{equation*}
L_{\widetilde H_{n_1}} \theta_{x}  (q,\omega)
\; = \; 
- \nabla_q \widetilde H \cdot \nabla_\omega \theta_{x} (q,\omega)
\; = \; 
- i \sum_{\modek\in K_r}  \big( \modek \cdot \nabla_\omega \theta_{x} (\omega) \big) \rho_\delta (\modek \cdot \omega) G(\modek,\omega) \ed^{i \modek\cdot q} .
\end{equation*}
It is thus enough to show that
\begin{equation*}
\modek \cdot \nabla_\omega \theta_{x} (\omega) = 0
\;\; \text{for every} \;\;
\modek \in K_r
\;\; \text{and every}\;\;
\omega \notin \mathsf S_{n_2} (x)
\;\; \text{such that} \;\;
|\modek \cdot \omega| \le 2 \delta.
\end{equation*}
Let us thus fix $\omega\in \R^{N}$ and $\modek \in K_r$ with these restrictions.
By definition \eqref{smooth indicator good set}, 
we see that $ \modek \cdot \nabla_\omega \theta_{x} (\omega) = 0$ if, for every $\omega'\in \R^{N}$ such that $\max_x|\omega'_x| \le 4\delta$, 
it holds that 
\begin{equation*}
\omega + \omega' \in \mathsf R_{n_2} (x)
\quad \Rightarrow \quad
\omega + \omega' + t \modek \in \mathsf R_{n_2} (x)
\;\; \text{for all $t$ such that $|t|$ is small enough}.
\end{equation*}
Here, the maximal value allowed for $|t|$ may depend on $\omega$ but not on $\omega'$.
We distinguish three cases: either at least one of the cases 1 and 2  is realized, or, if none of them is realized, than case 3 is.

1. 
There exists a cluster $\{ \modek_1 , \dots , \modek_p \}$ around $x$, with $p\le n_2$, such that 
$\omega + \omega' \in \mathsf B (\modek_1 , \dots , \modek_p)$ and that 
$\modek\,\bot\,\spanv \{ \modek_1 , \dots ,\modek_p \}$.
It is then seen from the definition of $\mathsf B (\modek_1 , \dots , \modek_p)$ that, 
for every $t\in\R$, $\omega + \omega' + t \modek \in  \mathsf B (\modek_1 , \dots , \modek_p)$. 
Therefore $\omega + \omega' + t \modek \in \mathsf R_{n_2} (x)$ for every $t\in\R$.

2.
There exists a cluster $\{ \modek_1 , \dots , \modek_p \}$ around $x$, with $p\le n_2$, such that 
$\omega + \omega' \in \mathsf B (\modek_1 , \dots , \modek_p)$ and that 
$\modek\in\spanv \{ \modek_1 , \dots , \modek_p \}$.
Since $|\modek \cdot \omega | \le 2 \delta$ and since $\max_x|\omega_x| \le 4\delta$, it holds that $|\modek \cdot (\omega + \omega')| \le (4r^2 + 2) \delta$.
Then, by Lemma \ref{lem: invariance of the sets B}, for $|t|\le \delta$ we still have $\omega + \omega' + t \modek \in \mathsf B (\modek_1 , \dots , \modek_p)$
if $L$ was chosen large enough. 
Therefore $\omega + \omega' + t \modek \in \mathsf R_{n_2} (x)$ for $|t|\le \delta$. 

3. 
For any cluster $\{ \modek_1 , \dots ,\modek_p \}$ around $x$, with $p\le n_2$, such that 
$\omega + \omega' \in \mathsf B (\modek_1 , \dots , \modek_p)$, it holds that 
$\modek\,\notin\,\spanv \{ \modek_1 , \dots ,\modek_p \}$, 
and that $\modek\,\cancel{\bot}\,\spanv \{ \modek_1 , \dots , \modek_p \}$. 
Let us see that, since we assume that $\omega \notin \mathsf S_{n_2}(x)$, this case actually does not happen. 
First, for all these clusters, we should have $p=n_2$.
Indeed, otherwise $\{ \modek_1 , \dots , \modek_p, \modek \}$ would form a cluster around $x$ containing $p+1 \le n_2$ independent vectors. 
We would then conclude as in case 2 that  $|\modek \cdot (\omega + \omega')| \le (4r^2 + 2) \delta$,
so that, by Lemma \ref{lem: extension B}, 
$\omega + \omega' \in \mathsf B(\modek_1, \dots ,\modek_p, \modek)$ if $L$ has been chosen large enough. 
This would contradict the assumption ensuring that we are in case 3. 
So $p=n_2$ should hold. 
Writing $\omega'' = \omega + \omega'$, 
we should then conclude from the definition of $\mathsf B (\modek_1 , \dots , \modek_p)$ that, for $1 \le j \le n_2$, 
\begin{align*}
|\modek_j \cdot \omega''|
\; &= \;
|\modek_j|_2
\big| \omega'' - P (\omega'',\pi(\modek_j)) \big|_2
\; \le \; 
|\modek_j|_2
\big| \omega'' - P \big(\omega'',\pi(\modek_1) \cap \dots \cap \pi(\modek_{n_2}) \big) \big|_2\\
\; &\le \; 
|\modek_j|_2 L^{n_2}\delta
\end{align*}
But then 
\begin{align*}
|\modek_j \cdot \omega|
\; &= \; 
|\modek_j \cdot (\omega + \omega') - \modek_j \cdot \omega'|
\; \le \; 
|\modek_j \cdot \omega''| + |\modek_j \cdot \omega'|
\; \le \; 
|\modek_j|_2 L^{n_2}\delta + 4 r^2 \delta\\
\; &\le \; 
L^{n_2 + 1} \delta
\end{align*}
if $L$ is large enough. 
This would contradict $\omega \notin \mathsf S_{n_2}(x)$.
\end{proof}

\section{Proof of Theorem \ref{the: decomposition fo the current}: the rotor chain}\label{sec: proof of the principal theorem}

Let $a \in \Z_N$ be given by hypothesis.
Let us assume that the dynamics is generated by Hamiltonian \eqref{Hamiltonian Rotors}.

\subsection{New decomposition of the Hamiltonian}

The original decomposition of the Hamiltonian leading to the definition of the current $\epsilon J_{a,a+1}$ is given by 
\begin{equation}\label{original decomposition of the original hamiltonian}
H 
\; = \; 
H_{\le a}^{\mathrm O} + H_{>a}^{\mathrm O}
\; = \; 
\sum_{x \le a} H_x \; + \sum_{x > a} H_x. 
\end{equation}
We will now obtain a new decomposition of the Hamiltonian 
that is  equivalent to the one above from the point of view of the conductivity, 
but leading to an instantaneous current that vanishes for most of the configurations in the Gibbs state at temperature $T$.  
 
Let $n_3 \ge 1$. 
For $x\in \mathrm B(a,n_3)$, we define
\begin{align}
\vartheta_{a,x}
\; &= \;   \frac{1}{\caN}   \bigg(  
\Big(\prod_{y \in \mathrm B(a,n_3)} \theta_y\Big)\,  \delta_{a,x} \, + \, \Big(1- \prod_{y \in \mathrm B(a,n_3)} \theta_y \Big) \,  \theta_x    
\bigg), 
\label{definition of vartheta a x}\\[2mm]
\vartheta_{a,*}
\; &= \;    \frac{1}{\caN}    \prod_{y\in\mathrm B(a,n_3)}  (1 - \theta_y).
\label{definition of vartheta a *}
\end{align}
with the normalization factor 
\beq
\caN \; = \; 
\Big(\prod_{y \in \mathrm B(a,n_3)} \theta_y\Big) \, +  \,  
\Big(1-\prod_{y \in \mathrm B(a,n_3)} \theta_y\Big)\Big(\sum_{x \in \mathrm B(a,n_3)} \theta_x \Big) 
\, + \, 
\prod_{y\in\mathrm B(a,n_3)}  (1 - \theta_y) 
\eeq
chosen so that
\begin{equation*}
\sum_{x\in\mathrm B(a,n_3)} \vartheta_{a,x} \; + \; \vartheta_{a,*} \; = \; 1,
\end{equation*}
and satisfying $\caN\geq 1$.   
We then define
\begin{align}
\widetilde{H}_{\le a}
\; &= \; 
\sum_{x\in \mathrm B(a,n_3)} \vartheta_{a,x} \sum_{y \le x} \widetilde H_y
\; + \;
\vartheta_{a,*} \sum_{y \le a} \widetilde H_y, 
\label{H tilde less than a}\\
\widetilde{H}_{> a}
\; &= \; 
\sum_{x\in \mathrm B(a,n_3)} \vartheta_{a,x} \sum_{y > x} \widetilde H_y
\; + \;
\vartheta_{a,*} \sum_{y > a} \widetilde H_y.
\label{H tilde bigger than a}
\end{align}
It holds that 
\begin{equation*}
\widetilde{H} \; = \; \widetilde{H}_{\le a} + \widetilde{H}_{>a}.
\end{equation*}
By the first point of Proposition \ref{pro: change of variable}, we finally define a new decomposition
\begin{equation}\label{new decomposition of the original hamiltonian}
H 
\; = \; 
H_{\le a} + H_{>a}
\; = \; 
\mathcal T_{n_1} (R \widetilde{H}_{\le a}) + \mathcal T_{n_1} (R \widetilde{H}_{> a}).
\end{equation}

\subsection{Definition of $U_a$ and $G_a$}

With the definitions \eqref{original decomposition of the original hamiltonian} and \eqref{new decomposition of the original hamiltonian}, 
and applying the second point of Proposition \ref{pro: change of variable}, we find that
\begin{align}
\epsilon J_{a,a+1} 
\; &= \; 
L_H H_{> a}^{\mathrm O} 
\; = \;
L_H (H_{> a}^{\mathrm O}  - H_{> a}) + L_H H_{> a} 
\label{redecoupe du courant}\\
\; &= \; 
L_H (H_{> a}^{\mathrm O}  - H_{> a}) 
+
\mathcal T_{n_1} (R L_{\widetilde{H}} \widetilde{H}_{> a}) 
+ 
\epsilon^{n_1+1} L_V \sum_{k=0}^{n_1} R^{(n_1-k)} \widetilde{H}_{> a}^{(k)}.
\end{align}
Let us call $n_0$ the number $n$ appearing in the statement of the Theorem. 
We define
\begin{align}
U_a 
\; &= \; 
H_{> a}^{\mathrm O}  - H_{> a} \; - \; \langle H_{> a}^{\mathrm O}  - H_{> a} \rangle_T , 
\label{definition of U a}\\
\epsilon^{n_0+1} G_a 
\; &= \; 
\mathcal T_{n_1} (R L_{\widetilde{H}} \widetilde{H}_{> a}) 
+ 
\epsilon^{n_1+1} L_V \sum_{k=0}^{n_1} R^{(n_1-k)} \widetilde{H}_{> a}^{(k)}.
\label{definition fo G a}
\end{align}
We notice that $\langle G_a \rangle_T = 0$ since $\epsilon^{n_0 + 1} \langle G_a \rangle_T = \langle L_H H_{>a} \rangle_T = 0$, by invariance of the Gibbs state.

\subsection{Locality}
Let us show that the functions $U_a$ and $G_a$ are local, meaning that they depend only on variables indexed by $z$
with $|z - a| \le \mathrm C_{n_0}$, for some constant $\mathrm C_{n_0} < + \infty$. 
To study $U_a$ we observe that 
\begin{equation*}
H_{>a}^{\mathrm O} - H_{>a} 
\; = \;
- ( H_{\le a}^{\mathrm O} - H_{\le a}). 
\end{equation*}
Let us see that $H_{>a}^{\mathrm O} - H_{>a}$ depends only on variables indexed by $z$ with $z \ge a -  (n_3 + (n_2 + 5)r)$.
The function $H_{>a}^{\mathrm O}$ only depends on variables indexed by $z$ with $z \ge a$. 
To analyze $H_{>a}$ defined by \eqref{new decomposition of the original hamiltonian}, 
we first notice that the functions $\vartheta_{a,x}$, with $x\in \mathrm B (a,n_3)$, and $\vartheta_{a,*}$, defined by (\ref{definition of vartheta a x}-\ref{definition of vartheta a *}),
only depend on variables indexed by $z$ with $z \ge a - (n_3 + 4 r + n_2 r)$.
By \eqref{H tilde bigger than a}, the same holds true for $\widetilde{H}_{>a}$, 
since, for any $x\in \Z_N$, the functions $\widetilde H_x$ only depend on variables indexed by $z$ with $z\ge x-r$. 
By \eqref{dependence on r in change of variables}, we conclude that $R\widetilde{H}_{>a}$, and so $H_{>a}$, only depends on variables indexed by $z$ with $z \ge a - (n_3 + 4r + n_2 r + r)$. 
The same holds thus for $H_{>a}^{\mathrm O} - H_{>a}$. 
We could similarly show that $ H_{\le a}^{\mathrm O} - H_{\le a}$ only depends on variables indexed by $z$ with $z \le a +   (n_3 + (n_2 + 5)r)$. 
We conclude that $U_a$ defined by \eqref{definition of U a} only depends on variables indexed by $z$ with $|z - a| \le  (n_3 + (n_2 + 5)r)$. 

We then readily conclude that $G_a$ is local as well, since, 
going back to \eqref{redecoupe du courant}, we see that $\epsilon^{n_0+1} G_a$ is the sum of two local functions:
\begin{equation*}
\epsilon^{n_0+1} G_a 
\; = \; 
L_H H_{> a} 
\; = \; 
\epsilon J_{a,a+1} - L_H (H_{>a}^{\mathrm O} - H_{>a}).
\end{equation*}

\subsection{An expression for $L_{\widetilde{H}} \widetilde{H}_{>a}$} \label{sec: expression for commutator}

We have
\begin{align}
L_{\widetilde H}  \widetilde{H}_{>a}
\; =& \; 
\sum_{x\in \mathrm B(a,n_3)} \vartheta_{a,x} \cdot \bigg( L_{\widetilde H} \sum_{y > x} \widetilde H_y \bigg)
\; + \; 
\vartheta_{a,*} \cdot\bigg( L_{\widetilde H} \sum_{y > a} \widetilde H_y \bigg) 
\nonumber\\
& \; + \; 
\sum_{x\in \mathrm B(a,n_3)} \big( L_{\widetilde H}\vartheta_{a,x} \big) \sum_{y > x} {\widetilde H}_y
\; + \;
\big( L_{\widetilde H} \vartheta_{a,*} \big) \sum_{y > a} \widetilde{H}_y .
\label{first expression for the main commutator}
\end{align}
Let us show that the terms in the first sum in the right hand side vanish, i.e.\
\begin{equation}\label{crucial cancellation}
\vartheta_{a,x} \cdot\bigg( L_{\widetilde H} \sum_{y > x} \widetilde H_y \bigg) \; = \; 0
\quad \text{for all} \quad x\in\mathrm B(a,n_3).
\end{equation}
Thanks to the presence of the operator $\mathcal R$ in \eqref{definition of H tilde}, and thanks to \eqref{dependence on r in change of variables}, 
we decompose $\widetilde H=\sum_x \widetilde H_x$ where $\widetilde H_x$ takes the form
\begin{align} \label{eq: splitting D and V}
\widetilde{H}_x (q,\omega)
\; &= \;
\sum_{\substack{\modek \in K_r: \\ \supp \modek \subset \mathrm B(x,r)}} \rho_\delta (\modek \cdot \omega) \widetilde{V}_x(\modek, \omega) \ed^{i\modek\cdot q} \nonumber\\
\; &= \;
 \widetilde D_x (\omega)  \; + \;  \sum_{\substack{\modek\in K_r: \modek\ne 0 \\ \supp \modek \subset \mathrm B(x,r)}} \rho_\delta (\modek \cdot \omega) \widetilde{V}_x(\modek, \omega) \ed^{i\modek\cdot q}, 
\end{align}
where we have singled out the $\modek=0$ mode by setting $ \widetilde D_x (\omega) = \widetilde{V}_x(0, \omega)$. We note that
 $\widetilde D_x$  depends only on variables indexed by $z$ with $z \in \mathrm B(x,r)$. 

Using the more handy notation $L_f g = \{ f,g\}$, we compute
\begin{align}
L_{\widetilde H} \sum_{y > x} \widetilde H_y
\; = & \; 
\Big\{ 
\sum_{z \le x} \widetilde H_z , \sum_{y > x} \widetilde H_y
\Big\}   \; = \; 
 \sum_{z \le x < y: | z-y | \le 2r} \Big\{ 
 \widetilde H_z ,  \widetilde H_y
\Big\}  \label{eq: towards crucial cancellation}
\end{align} 
where we have used the fact  that the Poisson bracket of two functions depending on variables indexed by points belonging to different, non-intersecting, subsets of $\Z_N$, vanishes.    Relation \eqref{crucial cancellation} surely holds if $\omega$ is such that $\vartheta_{a,x} (\omega) = 0$. 
Let us thus take $\omega$ such that $\vartheta_{a,x} (\omega) > 0$, which implies $\theta_x (\omega) > 0$.
It then follows from the first point of Proposition \ref{pro: invariance resonant set} 
that all the factors $\rho_\delta (\modek \cdot \omega)$ appearing when the operators $ \widetilde H_z ,  \widetilde H_y$ in \eqref{eq: towards crucial cancellation}   are written as in \eqref{eq: splitting D and V}, vanish.  Therefore,  when  $\vartheta_{a,x} (\omega) > 0$, the expression \eqref{eq: towards crucial cancellation} equals
\begin{equation*}
\sum_{z \le x < y: \str z-y \str \le 2r} \Big\{ 
\widetilde D_z ,  \widetilde D_y
\Big\} 
\end{equation*}
However, since $\widetilde D_z, \widetilde D_y$ depend only on the $\om$-variables, their Poisson bracket vanishes.
This establishes \eqref{crucial cancellation}.

Thanks to \eqref{crucial cancellation}, the Poisson bracket \eqref{first expression for the main commutator} is rewritten as
\begin{multline}\label{expression for L H tilde after cancellation}
L_{\widetilde H}  \widetilde{H}_{>a}
\; = \;
\vartheta_{a,*} \,
\Big\{
\sum_{a-r \le x \le a} \widetilde{H}_x \, , \sum_{a < y \le a+r} \widetilde H_y
\Big\} \\
\; + \; 
\sum_{x\in \mathrm B(a,n_3)} \big( L_{\widetilde H}\vartheta_{a,x} \big) \sum_{x < y \le a + n_3} {\widetilde H}_y
\; + \;
\big( L_{\widetilde H} \vartheta_{a,*} \big) \sum_{a < y \le a+ n_3} \widetilde{H}_y .
\end{multline}

\subsection{Definition of an exceptional set $\mathsf Z\subset \Omega$}
  
Let $\mathsf Z  \subset \Omega$ be such that $(\omega ,q) \in \mathsf Z$ if and only if
theres exists $n_2$ linearly independent vectors $\modek_1, \dots , \modek_{n_2} \in K_r$ such that $( \cup_j \supp (\modek_j) ) \subset \mathrm B (a,2n_3)$
and such that $|\modek_j \cdot \omega | \le L^{n_2 + 1} \delta$ for $1\le j \le n_2$. 
The set $\mathsf Z$ is closed. 

\begin{Lemma}\label{Lem: exceptional set Z}
If, given $n_1$ and $n_2$, the numbers $L$ and $n_3$ have been taken large enough, then
\begin{enumerate}
\item
If $L_{\widetilde H}  \widetilde{H}_{>a} (\omega , q) \ne 0$, then $(\omega, q) \in \mathsf Z$. 
\item
There exists $\mathrm C = \mathrm C(L,n_1,n_2,n_3) < + \infty$ such that 
$\langle \chi_\mathsf Z \rangle_T  \le \mathrm C \delta^{n_2}$.
\end{enumerate}
\end{Lemma}

\begin{proof}
Let us start with the first point.
From \eqref{expression for L H tilde after cancellation}, we conclude that, 
if $L_{\widetilde H}  \widetilde{H}_{>a} (\omega , q) \ne 0$, 
then at least one of the following quantities needs to be non-zero: 
$\vartheta_{a,*}  (\omega)$, 
or $L_{\widetilde H} \vartheta_{a,*}(\omega,q)$, 
or $L_{\widetilde H}\vartheta_{a,x} (\omega,q)$ for some $x\in \mathrm B(a,n_3)$. 
In fact, since $0 \le \prod_{x\in\mathrm B (a,n_3)} (1 - \theta_x) \le 1$, and since $L_{\widetilde H}$ is a differential operator, 
$L_{\widetilde H} \prod_{x\in\mathrm B (a,n_3)} (1 - \theta_x) (\omega,q) \ne 0$ implies $\prod_{x\in\mathrm B (a,n_3)} (1 - \theta_x)  (\omega) \ne 0$.
Therefore, by inspection of the definitions \eqref{definition of vartheta a x} and \eqref{definition of vartheta a *}, the condition $L_{\widetilde H}  \widetilde{H}_{>a} (\omega , q) \ne 0$ implies actually 
\begin{equation*}
\prod_{x\in\mathrm B (a,n_3)} (1 - \theta_x) \ne 0 
\qquad \text{or} \qquad 
L_{\widetilde H}\theta_{x} (\omega,q) \ne 0 
\quad \text{for some} \quad  
x\in \mathrm B(a,n_3). 
\end{equation*}

If $\prod_{x\in\mathrm B (a,n_3)} (1 - \theta_x) \ne 0 $, then $\theta_x (\omega) < 1$ for all $x \in \mathrm B(a,n_3)$. 
If $\theta_x (\omega) < 1$, there exists then, by the definition \eqref{smooth indicator good set}, some $\omega' \in \R^N$ such that $\max_y |\omega_y'| \le 2 \delta$, 
and such that $\omega'' = \omega + \omega' \in \mathsf R (x)$.
There exists therefore a cluster $\{\modek_1, \dots , \modek_p \}$ around $x$, with $p \leq n_2$, such that \eqref{first condition B set} holds. 
This implies 
\begin{align*}
|\modek_1 \cdot \omega''|
\; &= \;
|\modek_1|_2
\big| \omega'' - P (\omega'',\pi(\modek_1)) \big|_2
\; \le \; 
|\modek_1|_2
\big| \omega'' - P \big(\omega'',\pi(\modek_1) \cap \dots \cap \pi(\modek_{p}) \big) \big|_2\\
\; &\le \; 
|\modek_1|_2
L^{p}\delta
\end{align*}
and therefore $|\omega \cdot \modek_1| \le |\modek_1|_2 L^{p}\delta + r^2 \delta \le L^{n_2+1} \delta$ if $L$ is large enough and using that $p \leq n_2$. 
It holds by definition of a cluster around $x$ that $\supp(\modek_1) \subset \mathrm B(x,4r)$.   Let us now take another $x'$ such that $\theta_{x'} >0$ and $\str x-x'\str >4r$. Then the same reasoning gives a vector $\modek'_1 \neq \modek_1$ satisfying again  $|\omega \cdot \modek'_1|  \le L^{n+1} \delta$.   By taking   $n_3$ large enough, we can find $n_2$ linearly independent vectors and thus  guarantee that $\omega\in \mathsf Z$.

Suppose now that $L_{\widetilde H}\theta_{x} (\omega,q) \ne 0$ for some $x\in \mathrm B(a,n_3)$.  
It then follows from the second assertion of Proposition \ref{pro: invariance resonant set} that $\omega \in \mathsf S(x)$, 
so that, by definition, there exists a cluster $\{ \modek_1, \dots , \modek_{n_2} \}$ around $x$ such that $|\omega \cdot k_j| \le L^{n_2 + 1} \delta$. 
This implies $\omega \in \mathsf Z$.

We now move to the second claim of the Lemma.
Since the function $\chi_\mathsf{Z}$ depends only on the variables $\omega_x$ with $x \in \mathrm B(a,2n_3)$, 
and since the Gibbs measure factorizes with respect to the variables $\omega_y$ ($y\in \Z_N$), we get
\begin{equation*}
\langle \chi_{\mathsf{Z}} \rangle_T \; = \; 
\frac{\int \chi_{\mathsf Z} (\omega) \prod_{x\in \mathsf B(a, 2n_3)}  \ed^{- \omega_x^2/T} \, \dd \omega_x }{\int \prod_{x\in \mathsf B(a, 2n_3)}  \ed^{- \omega_x^2/T} \, \dd \omega_x}
\; \le \; 
\mathrm C (n_3) \int \chi_\mathsf Z (\omega)  \prod_{x\in \mathsf B(a, 2n_3)} \ed^{- \omega_x^2/T}  \, \dd \omega_x.
\end{equation*}
The result follows by a straightforward computation that exploits that the set $\mathsf Z$ is determined by $n_2$ constraints. 
\end{proof}

\subsection{Bounds on the norms of $U_a$ and $G_a$}

Let $\partial_\sharp$ be the derivative with respect to any $\omega_z$ or $q_z$ with $z\in \Z_N$, or even no derivative at all ($\partial_\sharp f = f$).
To lighten some notations, we again use the shorthand
$g \sim \delta^{-n}$ with the same meaning as in the proof of \eqref{H depends on delta} and \eqref{R depends on delta}.
We now allow $\delta$ to depend on $\epsilon$, and we fix the values of $n_1$ and $n_2$. 

Let us first obtain the bounds  $ \langle U^2_a \rangle_T \le \mathrm C_{n_0}\ep^{1/4}, \langle (\partial_\sharp U_a)^2 \rangle_T \le \mathrm C_{n_0}\ep^{-1/4}$ 
(recall that  we denote by $n_0$ the number $n$ appearing in Theorem \ref{the: decomposition fo the current}). 
We use \eqref{H depends on delta} and \eqref{R depends on delta}, 
and the fact that the functions $\vartheta_{a,x},\vartheta_{a,*}$ are bounded, to see that the local function $H^\mathrm O_{> a} - H_{> a}$ takes the form 
\begin{align}
H^\mathrm O_{> a} - H_{> a}
\; &= \; \sum_{n=0}^{n_1} \ep^{n} (H^\mathrm O_{> a} - H_{> a})^{(n)} \label{decomposition  H original - H in powers of epsilon} \\
\; &= \; 
H^\mathrm O_{> a} - \mathcal T_{n_1} \big( R \widetilde H_{>a}\big)
\; = \;
H^\mathrm O_{> a} - \sum_{n=0}^{n_1} \epsilon^n \sum_{k=0}^n R^{(n-k)} \widetilde H^{(k)}_{> a}
\; \sim \;
\sum_{n=0}^{n_1} \epsilon^{n} \delta^{-2n} \label{decomposition  H original - H differently}
\end{align}
with the dominant contribution for each $n$ coming from the $k=0$ term. 
Let us fix $\delta = \epsilon^{1/4}$. 
Let us start with the term corresponding to $n=0$ in \eqref{decomposition  H original - H in powers of epsilon}. 
From \eqref{decomposition  H original - H differently} we have
\begin{align*}
(H^\mathrm O_{> a} - H_{> a})^{(0)} & = 
\sum_{x >a} D_x -  \tilde H^{(0)}_{>a}  \\
 &=   \sum_{x >a} D_x -   \sum_{x\in \mathrm B(a,n_3)}  \vartheta_{a,x}  \sum_{y >x}  D_y -  \vartheta_{a,*} \sum_{y > a} D_y
\end{align*}
Let  $\mathsf W \subset \Omega$ be the set containing all $(\om,q)$ such that $\theta_x(\omega) <1$ for some $x\in \mathrm B(a,n_3)$. 
By inspection of the definitions \eqref{definition of vartheta a x} and \eqref{definition of vartheta a *}, we have
\begin{equation*}
\vartheta_{a,x}(\om)=\delta_{a,x}, \quad  \vartheta_{a,*}(\om)=0, \quad  (H^\mathrm O_{> a} - H_{> a})^{(0)}(\om,q) =0, \quad \text{for}\, \,  (\om,q) \in \Omega \setminus \mathsf W.
\end{equation*}
Therefore $(H^\mathrm O_{> a} - H_{> a})^{(0)} = \chi_{\mathsf W} \cdot (H^\mathrm O_{> a} - H_{> a})^{(0)} $.
So, arguing as in the proof of Lemma \ref{Lem: exceptional set Z}, we find that, 
since $\chi_{\mathsf W} \cdot (H^\mathrm O_{> a} - H_{> a})^{(0)} $ depends ony on the variables $\omega_x$, with $|x-a| \le \mathrm C $,
\begin{align}
\left\langle \left( (H^\mathrm O_{> a} - H_{> a})^{(0)} \right)^2 \right\rangle_T
\; &\le \; 
\mathrm C  \int \chi_{\mathsf W} (\omega)  \left( (H^\mathrm O_{> a} - H_{> a})^{(0)} \right)^2 (\omega) \prod_{x:|x-a| \le \mathrm C} \ed^{-\omega_x^2/T} \,\dd \omega_x 
\nonumber\\
\; &\le \; 
\mathrm C' \int \chi_{\mathsf W} (\omega)   \prod_{x:|x-a| \le \mathrm C} \ed^{-\omega_x^2/2T} \,\dd \omega_x
\; \le \; 
\mathrm C'' \delta
\; = \; 
\mathrm C'' \epsilon^{1/4}, \label{the first bound obtained in the 5.6 proof}
\end{align}
where, to get the last inequality, we used that, for any  $(\omega,q) \in \mathsf W$, 
there is at least  one $\modek \in K_r$, with $\supp(\modek) \subset \mathrm B(a,n_3)$, such that  $|\omega \cdot \modek| \le  L^{n_2+1} \delta$. 
Similarly, since $\partial_\sharp \vartheta_{a,x}, \partial_\sharp \vartheta_{a,*} \sim \delta^{-1}$, we have
\begin{align}
\left\langle \left( \partial_\sharp (H^\mathrm O_{> a} - H_{> a})^{(0)} \right)^2 \right\rangle_T
\; &\le \; 
\mathrm C  \int \chi_{\mathsf W} (\omega)  \left( \partial_\sharp (H^\mathrm O_{> a} - H_{> a})^{(0)} \right)^2 (\omega) \prod_{x:|x-a| \le \mathrm C} \ed^{-\omega_x^2/T} \,\dd \omega_x 
\nonumber\\
\;& \le \; 
\mathrm C \frac{\delta}{\delta^2}  \; \le \; \mathrm C \epsilon^{-1/4}.  \label{similar to first bound in the 5.6 proof}
\end{align}
Next, we conclude from \eqref{polynomial bounds} and \eqref{decomposition  H original - H differently} that for $1 \le n \le n_1$, we have
\begin{equation} \label{the second bound obtained in the 5.6 proof}
\left\langle \left( \epsilon^n \partial_\sharp ( H^\mathrm O_{> a} - H_{> a})^{(n)}  \right)^2 \right\rangle_T 
\; \le \;  \mathrm C \left( \epsilon^n \delta^{-(2n+1)}\right)^2 \; \le \; \mathrm C \epsilon^{1/2}.
\end{equation}
Using (\ref{the first bound obtained in the 5.6 proof}-\ref{the second bound obtained in the 5.6 proof}) in \eqref{decomposition  H original - H in powers of epsilon}, we deduce
\begin{equation*}
\langle  (H^\mathrm O_{> a} - H_{> a})^2 \rangle_T  
\; \le \; 
\mathrm C \epsilon^{1/4},\qquad  \langle (\partial_\sharp (H^\mathrm O_{> a} - H_{> a}))^2 \rangle_T  \; \le \; \mathrm C \epsilon^{-1/4},
\end{equation*}
from which the claimed bounds on $U_a$ follow.

Next, to obtain the bound $\langle (\partial_\sharp G_a)^2 \rangle_T \le \mathrm C_{n_0}$, we start from the definition \eqref{definition fo G a} and we note that both terms in this definition are local;
for $\mathcal T_{n_1} (R L_{\widetilde{H}} \widetilde{H}_{> a}) $ this follows from the explicit expression in Section \ref{sec: expression for commutator} 
and the other term is then local as a difference of local terms. 
We compute
\begin{equation}\label{second brol preuve theoreme prinicpal}
\langle (\partial_\sharp G_a)^2 \rangle_T 
\; \le \; 
2\epsilon^{-2 (n_0 + 1)} \big\langle \big( \partial_\sharp \mathcal T_{n_1} (R L_{\widetilde{H}} \widetilde{H}_{> a})  \big)^2 \big\rangle_T
\; + \; 
2\epsilon^{2 (n_1 - n_0)} \Big\langle \Big(  \partial_\sharp L_V \sum_{k=0}^{n_1} R^{(n_1-k)} \widetilde{H}_{> a}^{(k)}  \Big)^2 \Big\rangle_T .
\end{equation}

We look first at the second term and conclude, by means of \eqref{H depends on delta} and \eqref{R depends on delta} that it is of the form 
\begin{equation*}
\partial_\sharp L_V \sum_{k=0}^{n_1} R^{(n_1-k)} \widetilde{H}_{> a}^{(k)}
\; \sim \; 
\delta^{-2n_1 - 2}.  
\end{equation*}
We thus obtain, using locality, 
\begin{equation*}
\epsilon^{2 (n_1 - n_0)} \Big\langle \Big(\partial_\sharp  L_V \sum_{k=0}^{n_1} R^{(n_1-k)} \widetilde{H}_{> a}^{(k)}  \Big)^2 \Big\rangle_T 
\; \le \; \mathrm C_{n_0}
\qquad \text{if} \qquad 
\delta = \epsilon^{1/4} 
\quad \text{and} \quad
n_1=2n_0+1.
\end{equation*}

We then analyze the first term in \eqref{second brol preuve theoreme prinicpal}.
By the first point of Lemma \ref{Lem: exceptional set Z}, the function $L_{\widetilde{H}} \widetilde{H}_{> a}$ vanishes on the open set $\Omega \setminus \mathsf Z$, 
so that $ \partial_\sharp \mathcal T_{n_1} (R L_{\widetilde{H}} \widetilde{H}_{> a}) $ vanishes on this set as well. 
Using Cauchy-Schwarz inequality, and the second point of Lemma \ref{Lem: exceptional set Z}, we conclude that 
\begin{equation*}
\big\langle \big( \partial_\sharp \mathcal T_{n_1} (R L_{\widetilde{H}} \widetilde{H}_{> a})  \big)^2 \big\rangle_T
\; = \;
\big\langle \big( \partial_\sharp \mathcal T_{n_1} (R L_{\widetilde{H}} \widetilde{H}_{> a})  \big)^2 \chi_\mathsf Z \big\rangle_T
\; \le \;
\mathrm C \delta^{n_2/2} \big\langle \big( \partial_\sharp \mathcal T_{n_1} (R L_{\widetilde{H}} \widetilde{H}_{> a})  \big)^4 \big\rangle_{T}^{1/2}.
\end{equation*}
Using (\ref{H depends on delta}, \ref{R depends on delta}), we find that
\begin{equation*}
\partial_\sharp\mathcal T_{n_1} (R L_{\widetilde{H}} \widetilde{H}_{> a}) 
\; = \; 
\partial_\sharp \sum_{n=0}^{n_1} \epsilon^n \sum_{m=0}^n R^{(n-m)} \sum_{j=0}^m \big\{ \widetilde H^{(m-j)} , \widetilde H_{>a}^{(j)} \big\}
\; \sim \; 
\sum_{n=1}^{n_1} \epsilon^n \delta^{-2n} 
\; \sim \; \delta^{2}, 
\end{equation*}
where we have used the fact that $\{ \widetilde H^{(0)}, \widetilde H_{>a}^{(0)} \} = 0$ to start the last sum at $n=1$. 
We conclude that $\delta = \epsilon^{1/4}$ guarantees that 
\begin{equation*}
\big\langle \big( \partial_\sharp \mathcal T_{n_1} (R L_{\widetilde{H}} \widetilde{H}_{> a})  \big)^2 \big\rangle_T
\; \le \; 
\mathrm C \delta^{n_2/2+4}  \; \le \;  \mathrm C \epsilon^{n_2/8+1}.
\end{equation*}
taking $n_2 = 16n_0 +8$, we conclude from \eqref{second brol preuve theoreme prinicpal} that  $\langle (\partial_\sharp G_a)^2 \rangle_T \le \mathrm C_{n_0}$. 
This shows Theorem  \ref{the: decomposition fo the current} in the case where $H$ is given by \eqref{Hamiltonian Rotors}.

\section{Proof of Theorem \ref{the: decomposition fo the current} for the NLS chain}\label{sec: adaptation to nls}

The proof of the theorem given for rotors can readily be taken over to the non-linear Schr\"odinger chain by a adapting to this case 
the definition of the space $\mathcal S (\Omega)$ given in Section \ref{sec: approximate change of variables}.

We simply view functions on $\Omega = \C^N$ as functions on $\R^{2N}$. 
We use the notation $\omega = (\omega_x)_{x\in \Z_N} = (|\psi_x|^2)_{x\in\Z_N}$, and we observe that now $\omega \in (\R_+)^N$.
We will say that $f\in \mathcal S(\Omega)$ if the following conditions are met for some $r(f)<+\infty$:
We first ask, as before, that $f$ is smooth, 
that $f$ and all its derivatives are of polynomial growth
and that $f$ can be expressed as a sum of local terms, each of which is indexed by $z$ in a ball or radius $r(f)$.
We next replace \eqref{f finite number of Fourier modes} by the requirement that $f$ takes the form
\begin{equation*}
f (\psi) 
\; = \;
\sum_{\modek \in \Z_N} \widehat{f}(\modek,\omega) \psi^{\modek}
\qquad \text{with} \qquad
\widehat{f} (\modek,\omega) \; = \; 0 
\quad \text{if} \quad \max_x |\modek_x| \;\ge\; r(f),
\end{equation*}
where we have used the notation
\begin{equation*}
\psi^\modek
\; = \; 
\prod_{x\in\Z_N} \widetilde{\psi}^{\modek_x}_x
\quad \text{with} \quad 
\widetilde{\psi}^{\modek_x}_x
\; = \;
\left\{
\begin{array}{lll}
\psi_x^{\modek_x} & \text{if} & \modek_x > 0, \\
1 &\text{if} & \modek_x = 0,\\
\overline{\psi}_x^{\modek_x} & \text{if} & \modek_x < 0.
\end{array}
\right. 
\end{equation*}
It is then checked that the formulas involving explicitly the decomposition of $f$ in its Fourier components 
is simply transposed by means of the substitution $\ed^{i \modek \cdot q} \mapsto \psi^\modek$.

All the definitions and conclusions of Section \ref{sec: approximate change of variables} remain valid,
whereas the third claim of Proposition \ref{pro: change of variable} can be omitted since it is only used to prove Theorem \ref{the: noisy dynamics} stated for the rotor chain only. 
A word of comment on the derivation of the equivalent of (\ref{H depends on delta},\ref{R depends on delta}) may be useful.  
The notation $g \sim \delta^{-n}$ is now understood to mean that $g$ is a sum of terms of the form
$\delta^{-n} b(\omega/\delta,\delta)f(\psi,\bar\psi)$ with $b$ being exactly as in Section \ref{sec: approximate change of variables} and $f \in \mathcal S (\Omega)$.   
The definition of $\mathcal R$ and of the solution of the equation $L_D u=f$ is  again simply taken over from Section \ref{sec: approximate change of variables}. 
Let us show that, as before, it holds that if $g \sim \delta^{-n}$, $h \sim \delta^{-m}$, then 
\begin{equation}\label{Lgh relation for DNLS}
L_g h \sim  \delta^{-m-n-1}.
\end{equation}
Writing a term in $g(\psi,\overline{\psi})$ as $\delta^{-n}b(\omega/\delta,\delta)f(\psi,\overline{\psi})$, 
and a term in $h(\psi,\overline{\psi})$ as $\delta^{-m}b'(\omega/\delta,\delta)f'(\psi,\overline{\psi})$, we find that 
\begin{equation*}
L_{bf} (b'f') \; = \; (L_b b') ff' + (L_bf')fb' + (L_fb')bf' + (L_f f') bb'. 
\end{equation*}
As $b$ and $b'$ are functions of the variable $\omega = (|\psi_x|^2)_{x\in \Z_N}$ alone, it holds that $L_b b' = 0$, and \eqref{Lgh relation for DNLS} follows.
The relations (\ref{H depends on delta},\ref{R depends on delta}) are then derived as in Section \ref{sec: approximate change of variables}.

The results of Section \ref{sec: resonant frequencies} can also be taken over without change. 
It is observed that the functions $\theta_x$, $x\in \Z_N$, can still be defined for $\omega \in \R^N$, 
even though only their restriction to $(\R_+)^N$ is needed. 

The arguments of Section \ref{sec: proof of the principal theorem} remain all valid as well, 
except for the claims relying on the fact that the Gibbs measure factorizes with respect to the variables $\omega_x$ ($x\in \Z_N$).
This includes the proof of the second claim of Lemma \ref{Lem: exceptional set Z}, 
the argument leading to (\ref{the first bound obtained in the 5.6 proof},\ref{similar to first bound in the 5.6 proof}), 
and the property that the polynomials that depend only on some fixed variables, are integrable with respect to the Gibbs measure, with bounds that do not depend on the volume.
We now prove that $\langle \chi_{\mathsf Z} \rangle_T \le \mathrm C \delta^{n_2}$ for the non-linear Schr\"odinger chain as well 
(the other properties are shown in the same way).  
Let us assume $T=1$ for simplicity of notation. 

We have
\begin{equation}\label{expression chi Z NLS chain}
\langle \chi_{\mathsf Z} \rangle_1
\; = \;
\frac{\int \chi_\mathsf Z \ed^{- H(\psi)} \, \dd \psi}{ \int  \ed^{- H(\psi)} \, \dd \psi}, 
\end{equation}
with $\dd \psi = \dd \Re(\psi) \dd \Im(\psi)$. 
The function $\chi_\mathsf Z$ only depends on the variables indexed by $x\in\mathrm B(a,2n_3)$,
and it is thus natural to factorize the integrals into three pieces. 
For the numerator, we simply use the fact that $\ed^{-\phi} \le 1$ for all $\phi\ge 0$, to obtain that 
\begin{multline*}
\int \chi_\mathsf Z \ed^{- H(\psi)} \, \dd \psi
\; \le \; \\
\int \prod_{x < a - 2n_3} \ed^{-H'_x (\psi)} \, \dd \psi_x
\, \cdot \,
\int \chi_\mathsf Z \prod_{x\in \mathrm B(a,2n_3)} \ed^{-H'_x (\psi)} \, \dd \psi_x
\, \cdot \,
\int \prod_{x> a + 2n_3} \ed^{-H'_x (\psi)} \, \dd \psi_x,
\end{multline*}  
where $H_x'$ differs from $H_x$ only for the boundary terms $x=a-2n_3-1$ and $x=a+2n_3$ : 
\begin{align*}
H'_{a - 2n_3 - 1} \; &= \; H_{a-2n_3-1} - \epsilon |\psi_{a-2n_3} - \psi_{a-2n_3 -1}|^2,\\
H'_{a + 2n_3} \; &=\; H_{a+2n_3} - \epsilon |\psi_{a+2n_3+1} - \psi_{a+2n_3}|^2.
\end{align*}
These definitions are chosen such that $H'_{a - 2n_3 - 1}$ does not depend on $\psi_{a - 2n_3 }$ and $H'_{a + 2n_3}$ does not depend on $\psi_{a + 2n_3+1}$.
The middle integral is estimated as in the case of rotors:
\begin{equation*}
\int \chi_\mathsf Z \prod_{x\in \mathrm B(a,2n_3)} \ed^{-H'_x (\psi)} \, \dd \psi_x
\; \le \; 
\int \chi_\mathsf Z \prod_{x\in \mathrm B(a,2n_3)}  \ed^{-\frac{1}{2}|\psi_x|^4} \, \dd \psi_x 
\; \le \; 
\mathrm C \, \delta^{n_2}.
\end{equation*}
We then similarly factorize the numerator in \eqref{expression chi Z NLS chain}, using this time the bound
$|\psi_x - \psi_{x+1}|^2 \le 2 (|\psi_x|^2 + |\psi_{x+1}|^2)$ for the boundary terms:
\begin{multline*}
\int  \ed^{- H(\psi)} \, \dd \psi
\; \ge \; \\ 
\int \prod_{x < a - 2n_3} \ed^{-H''_x (\psi)} \, \dd \psi_x
\, \cdot \,
\int \prod_{x\in \mathrm B(a,2n_3)} \ed^{-H''_x (\psi)} \, \dd \psi_x
\, \cdot \,
\int \prod_{x> a + 2n_3} \ed^{-H''_x (\psi)} \, \dd \psi_x
\end{multline*}
where $H''_x$ differs from $H_x$ only for the following boundary terms:
\begin{align*}
H''_{a - 2n_3 - 1} \; &= \; H_{a - 2n_3 - 1} - \epsilon  |\psi_{a-2n_3} - \psi_{a-2n_3 -1}|^2 + 2 \epsilon |\psi_{a - 2n_3 - 1}|^2,\\
H''_{a -2n_3} \; &= \; H_{a-2n_3} + 2 \epsilon |\psi_{a - 2n_3}|^2 
\end{align*}
as well as $H''_{a+2n_3}$ and $H''_{a+2n_3 +1}$ that are defined similarly. 
The middle integral is immediately seen to be bounded from below by a constant depending on $n_3$.

To finish the proof, it is thus enough to show that there exists a constant $\mathrm C < +\infty$ such that 
\begin{equation*}
\frac{\int \prod_{x < a - 2n_3} \ed^{-H'_x (\psi)} \, \dd \psi_x}{\int \prod_{x < a - 2n_3} \ed^{-H''_x (\psi)} \, \dd \psi_x}
\; \le \; \mathrm C
\qquad \text{and} \qquad 
\frac{\int \prod_{x > a + 2n_3} \ed^{-H'_x (\psi)} \, \dd \psi_x}{\int \prod_{x > a + 2n_3} \ed^{-H''_x (\psi)} \, \dd \psi_x}
\; \le \; \mathrm C.
\end{equation*}
Both cases are treated similarly, and we consider the second one only. 
Writing $b = a + 2n_3 + 1$, we have
\begin{equation*}
\frac{\int \prod_{x \ge b} \ed^{-H'_x (\psi)} \, \dd \psi_x}{\int \prod_{x \ge b} \ed^{-H''_x (\psi)} \, \dd \psi_x}
\; = \; 
\frac{\int \ed^{2 \epsilon |\psi_b|^2} \prod_{x \ge b} \ed^{-H''_x (\psi)} \, \dd \psi_x}{\int \prod_{x \ge b} \ed^{-H''_x (\psi)} \, \dd \psi_x}.
\end{equation*}
Brascamp-Lieb inequalities furnish a possible way to estimate this integral (see \cite{caf}).
Since the function $z \mapsto |z|^4$ is not strictly convex at origin, we need however to slightly modify the measure at $x=b$.
There exist constants $c,c'>0$ small enough such that
\begin{equation*}
\psi_b \; \mapsto \; \frac{1}{2}|\psi_b|^4 + c \Big(1 - \frac{1}{1 + |\psi_b|^2}\Big) - c' |\psi_b|^2
\end{equation*}
is convex on $\C$. 
Letting then
\begin{equation*}
H_b''' (\psi) \; = \; H_b'' + c \Big(1 - \frac{1}{1 + |\psi_b|^2}\Big)
\end{equation*}
and $H_x''' = H_x''$ for $x> b$, we have
\begin{equation*}
\frac{\int \prod_{x \ge b} \ed^{-H'_x (\psi)} \, \dd \psi_x}{\int \prod_{x \ge b} \ed^{-H''_x (\psi)} \, \dd \psi_x}
\; \le \; 
\mathrm C \frac{\int \ed^{2 \epsilon |\psi_b|^2} \prod_{x \ge b} \ed^{-H'''_x (\psi)} \, \dd \psi_x}{\int \prod_{x \ge b} \ed^{-H'''_x (\psi)} \, \dd \psi_x}.
\end{equation*}
It is checked that 
\begin{equation*}
\psi \; \mapsto \; \sum_{x=b}^{(N-1)/2} H'''_x (\psi) - \Big( c' |\psi_b|^2 + \frac{\epsilon}{2}\sum_{x=b}^{(N-1)/2} |\psi_{x+1} - \psi_x|^2 \Big)
\end{equation*}
is convex, with the convention that $\psi_{(N+1)/2} = \psi_{(N-1)/2}$.
By Brascamp-Lieb inequality (Corollary 7 in \cite{caf}),
followed by the change of variables $\phi_b = \psi_b$ and $\phi_{x+1} = \psi_{x+1} - \psi_x$ for $b \le x \le (N-3)/2$, 
we conclude that 
\begin{align*}
\frac{\int \prod_{x \ge b} \ed^{-H'_x (\psi)} \, \dd \psi_x}{\int \prod_{x \ge b} \ed^{-H''_x (\psi)} \, \dd \psi_x}
\; &\le \; 
\mathrm C \,
\frac{\int \ed^{ -( c' - 2 \epsilon) |\psi_b|^2} 
\prod_{x \ge b} \ed^{-\frac{\epsilon}{2}|\psi_{x+1} - \psi_x|^2} \, \dd \psi_x}{\int \ed^{-c'|\psi_b|^2}\prod_{x \ge b} \ed^{-\frac{\epsilon}{2}|\psi_{x+1} - \psi_x|^2} \, \dd \psi_x}
\qquad (\psi_{(N+1)/2} = \psi_{(N-1)/2}) \\
& = \; 
\mathrm C \,
\frac{\int \ed^{ -( c' - 2 \epsilon) |\phi_b|^2} 
\prod_{x \ge b} \ed^{-\frac{\epsilon}{2}|\phi_x|^2} \, \dd \phi_x}{\int  \ed^{-c'|\phi_b|^2}\prod_{x \ge b} \ed^{-\frac{\epsilon}{2}|\phi_x|^2} \, \dd \phi_x} \\
& = \; 
\mathrm C \,
\frac{\int \ed^{ -( c' - 2 \epsilon) |\phi_b|^2}  \dd \phi_b}{\int \ed^{-c'|\phi_b|^2} \dd \phi_b}
\; \le \; \mathrm C'
\end{align*}
for some $\mathrm C' < + \infty$. 
$\square$

\section{Proofs of Theorems  \ref{the: weak coupling conductivity},  \ref{the: noisy dynamics} and \ref{the: nekoroshev} }\label{sec: last section}

The proof of Theorems \ref{the: weak coupling conductivity} and \ref{the: noisy dynamics} closely follows the proof of analog results in \cite{huv}, itself inspired by \cite{liv}.
We will need some decorrelation properties of the Gibbs measure. 
General results in \cite{led} apply to the measures corresponding to the Hamiltonians \eqref{Hamiltonian Rotors} and \eqref{Hamiltonian DNLS}, 
if $\epsilon$ is small enough for a given temperature $T$.
Given $A,B\subset \Z_N$, let $d(A,B) = \min\{ |x-y| : x\in A, y \in B\}$. 
Given a local function $f$ on $\Omega$, let $S(f)\subset \Z_N$ be the set of points such that $f$ does only depend on variables indexed by points in $S(f)$.
There exist constants $\mathrm C < + \infty$ and $c > 0$ such that 
given two smooth functions $f$ and $g$ on $\Omega$ satisfying $\langle f \rangle_T = \langle g \rangle_T = 0$, 
it holds that 
\begin{equation}\label{decorrelation}
| \langle f g \rangle_T |
\; \le \; 
\mathrm C \ed^{-c d(S(f),S(g))}
\langle \str\nabla f\str^2 \rangle_T^{1/2} \langle \str\nabla g\str^2 \rangle_T^{1/2}
\end{equation}
where $\str\nabla f\str^2=\sum_{x \in \Z_N}  \big(\str \partial_{\omega_x} f \str^2+\str \partial_{q_x} f \str^2\big) $ for the rotor chain and  $\str\nabla f\str^2=\sum_{x \in \Z_N}  \big(\str \partial_{\psi_x} f \str^2+\str \partial_{\bar \psi_x} f \str^2\big) $ for the NLS chain.  
Strictly speaking, \eqref{decorrelation} is stated in \cite{led} only in the case where the one-site phase space is $\R$,
 but  the proof goes through without any changes in our case as well (our one-site phase space is $\T \times \R$, and $\C$ respectively) since the only genuine requirement is a Poincar\' e inequality for the one-site measure. 

\begin{proof}[Proof of Theorem  \ref{the: weak coupling conductivity}]
Applying Theorem \ref{the: decomposition fo the current}, we write 
\begin{align*}
\epsilon \int_0^{\epsilon^{-n}t} \mathcal J_N (X_\epsilon^s) \, \dd s
\; =& \; 
\frac{\epsilon}{\sqrt N} \sum_{a\in \Z_N}\int_0^{\epsilon^{-n}t}  J_{a,a+1} (X_\epsilon^s) \, \dd s \\
\; =& \;
\frac{1}{\sqrt N} \sum_{a\in \Z_N} \big( U_a (X_\epsilon^{\epsilon^{-n}t}) - U_a (X_\epsilon^{0}) \big)
+
\frac{\epsilon^{n+1}}{\sqrt N} \sum_{a\in \Z_N}\int_0^{\epsilon^{-n}t}  G_a (X_\epsilon^s) \, \dd s.
\end{align*}
Therefore
\begin{multline*}
\Big\langle \Big(
\epsilon \int_0^{\epsilon^{-n}t} \mathcal J_N (X_\epsilon^s) \, \dd s
\Big)^2 \Big\rangle_T
\; \le \; 
2\Big\langle \Big(
\frac{1}{\sqrt N} \sum_{a\in \Z_N} \big( U_a (X_\epsilon^{\epsilon^{-n}t}) - U_a (X_\epsilon^{0}) \big)
\Big)^2 \Big\rangle_T \\
+ \;
2\Big\langle \Big(
\frac{\epsilon^{n+1}}{\sqrt N} \sum_{a\in \Z_N}\int_0^{\epsilon^{-n}t}  G_a (X_\epsilon^s) \, \dd s
\Big)^2 \Big\rangle_T
\end{multline*}
We conclude by stationarity of the Gibbs measure, 
by the decorrelation inequality \eqref{decorrelation}, and by the bounds \eqref{bornes sur U et G}  that
\begin{equation*}
\Big\langle \Big(
\frac{1}{\sqrt N} \sum_{a\in \Z_N} \big( U_a (X_\epsilon^{\epsilon^{-n}t}) - U_a (X_\epsilon^{0}) \big)
\Big)^2 \Big\rangle_T
\; \le \; 
2\Big\langle \Big(
\frac{1}{\sqrt N} \sum_{a\in \Z_N} U_a
\Big)^2 \Big\rangle_T
\; \le \; \mathrm C \ep^{-1/2}.
\end{equation*}
Next, by Jensen's inequality, by the invariance of the Gibbs measure, 
by the decorrelation inequality \eqref{decorrelation}, and by the bounds \eqref{bornes sur U et G},
we have
\begin{align*}
\Big\langle \Big(
\frac{\epsilon^{n+1}}{\sqrt N} \sum_{a\in \Z_N}\int_0^{\epsilon^{-n}t}  G_a (X_\epsilon^s) \, \dd s
\Big)^2 \Big\rangle_T
\; &\le \; 
\epsilon^{2(n+1)} \epsilon^{-n} t \int_0^{\epsilon^{-n}t}
\Big\langle\Big(
\frac{1}{\sqrt N} \sum_{a\in\Z_N} G_a (X_\epsilon^s)
\Big)^2 \Big\rangle_T \, \dd s  \\
\; &= \; 
\epsilon^{2(n+1)} \epsilon^{-2n} t^2 
\Big\langle\Big(
\frac{1}{\sqrt N} \sum_{a\in\Z_N} G_a
\Big)^2 \Big\rangle_T
\; \le \; 
\mathrm C \epsilon^2 t^2.  
\end{align*}
Finally we obtain that 
\begin{equation*}
\epsilon^{-m}\Big\langle\Big( 
\frac{\epsilon}{\sqrt{\epsilon^{-n}t}} \int_0^{\epsilon^{-n}t} \mathcal J_N (X^s_\epsilon) \, \dd s
\Big)^2\Big\rangle_T
\; = \; 
\frac{\mathrm C\epsilon^{n-m}}{t} (\ep^{-1/2} + \epsilon^2 t^2).
\end{equation*}
Because $n-m > 0$,
this quantity goes to zero when taking successively the limits $N\rightarrow \infty$, $\epsilon \rightarrow 0$ and $t\rightarrow \infty$. 
\end{proof}

\begin{proof}[Proof of Theorem  \ref{the: noisy dynamics}]
Let us write $\mathsf E_T(\cdot)$ for $\langle \mathsf E (\cdot)\rangle_T$.
Theorem \ref{the: decomposition fo the current} implies that, for any $a\in \Z_N$,
\begin{equation*}
\epsilon J_{a,a+1} 
\; = \; 
L_H U_a + \epsilon^{n+1} S U_a + \epsilon^{n+1} G_a - \epsilon^{n+1} S U_a
\; = \; 
\mathcal L U_a + \epsilon^{n+1} G_a - \epsilon^{n+1} S U_a
\end{equation*}
where $\mathcal L$ is defined by \eqref{noisy generator}.
Since $G_a$ is local and antisymmetric under the exchange $\omega \mapsto -\omega$, 
there exists a local function $F_a$ that solves the Poisson equation $S F_a = G_a$ and inherits of the properties of $G_a$ (see Lemma 2 in \cite{huv}).
To simplify notations, let us write 
\begin{equation*}
\mathcal U_N \; = \; \frac{1}{\sqrt N} \sum_{a\in \Z_N} U_a, 
\qquad
\mathcal K_N \; = \; \frac{1}{\sqrt N} \sum_{a \in \Z_N} (F_A - U_a).
\end{equation*}
We find that
\begin{equation*}
\mathsf E_T \Big( \frac{\epsilon}{\sqrt t} \int_0^t \mathcal J_N (\mathcal X_\epsilon^s) \, \dd s \Big)^2
\; \le \; 
2\mathsf E_T \Big( \frac{1}{\sqrt t} \int_0^t \mathcal L \mathcal U_N (\mathcal X_\epsilon^s) \, \dd s \Big)^2
\; + \; 
2\epsilon^{2(n+1)} \mathsf E_T \Big( \frac{1}{\sqrt t} \int_0^t S \mathcal K_N (\mathcal X_\epsilon^s) \, \dd s \Big)^2.
\end{equation*}
The first term is written as a sum of the variance of a stationary martingale and a rest term:
\begin{align}
\mathsf E_T \Big( \frac{1}{\sqrt t} \int_0^t \mathcal L \mathcal U_a \Big)^2
\; &= \; 
\big\langle \mathcal U_N \cdot (-\epsilon^{n+1} S) \mathcal U_{N}) \big\rangle_T
\; + \; \frac{1}{t}E_T \big( \mathcal U_N (\mathcal X_\epsilon^t) - \mathcal U_N (\mathcal X_\epsilon^0) \big)^2
\nonumber\\
\; &\le \; 
\epsilon^{n+1} \big\langle \mathcal U_N \cdot (-S) \mathcal U_{N}) \big\rangle_T
\; + \; \frac{2\langle \mathcal U^2_N \rangle_T}{t}
\; \le \; \mathrm C \ep^{-1/2} (\epsilon^{n+1} + 1/t),
\label{noisy case bound 1}
\end{align} 
where the last bound has been obtained as in the proof of Theorem \ref{the: weak coupling conductivity}.
The second term involves a function $S\mathcal K$, that obviously lies in the image of the symmetric part $S$ of the generator $\mathcal L$. 
A classical bound \cite{kip} yields
\begin{align}
\epsilon^{2(n+1)} \mathsf E_T \Big( \frac{1}{\sqrt t} \int_0^t S \mathcal K_N (\mathcal X_\epsilon^s) \, \dd s \Big)^2
\; &\le \;
\mathrm C  \epsilon^{2(n+1)} \big\langle S\mathcal K_N \cdot (-\epsilon^{n+1}S)^{-1} S\mathcal K_N \big\rangle_T
\nonumber\\
\; &= \;
\mathrm C \epsilon^{n+1} \big\langle S\mathcal K_N \cdot \mathcal K_N \big\rangle_T
\; \le \; 
\mathrm C' \epsilon^{n+1/2}, 
\label{noisy case bound 2}
\end{align}
where the last bound has been obtained as in the proof of Theorem \ref{the: weak coupling conductivity}.
The theorem is obtained by taking the limit $N\rightarrow \infty$ and then the limit $t\rightarrow \infty$ in \eqref{noisy case bound 1} and \eqref{noisy case bound 2}.
\end{proof}

\begin{proof}[Proof of Theorem  \ref{the: nekoroshev}]
By the definition of the currents $J_{a,a+1}$ we get
\begin{equation*}
L_H(H_I) \; = \; \epsilon J_{a_1,a_1+1} -  \epsilon J_{a_2,a_2+1} 
\end{equation*}
so that, by integrating over time the statement of Theorem \ref{the: decomposition fo the current}, 
\begin{equation*}
H_I(X^t_\epsilon) -H_I 
\; = \; 
\sum_{j=1,2} (-1)^{j+1} \Big(U_{a_j}(X^t_\epsilon)- U_{a_j}  + \epsilon^{n+1} \int_0^t \dd s  \, G_{a_j}(X^s_\epsilon) \Big).
\end{equation*}
By invariance of the Gibbs state, we have $\langle (U_{a_j}(X^t_\epsilon))^2 \rangle_T=\langle U_{a_j}^2 \rangle_T$ for $j=1,2$. 
Hence by analogous manipulations as those in  the proof of Theorem \ref{the: weak coupling conductivity}, we get 
\begin{equation*}
\langle ( H_I(X^t_\epsilon) -H_I)^2 \rangle_T \;\leq\;  \mathrm C(n) \sum_{j=1,2} \left( \langle U_{a_j}^2 \rangle_T \, + \, \epsilon^{2n+1}  t^2   \langle G_{a_j}^2 \rangle_T \right).
\end{equation*}
The theorem now follows by the bounds on $U_a,G_a$ stated in Theorem \ref{the: decomposition fo the current}, upon taking $t=\epsilon^{-n}$.  
\end{proof}

\paragraph{Acknowledgements:}
We thank the anonymous referee for pointing out several inaccuracies in our original manuscript. 
W.D.R thanks the DFG (German Research Fund) and the Belgian Interuniversity Attraction Pole  (P07/18 Dygest) for financial support. 
F.H. thanks the University of Helsinki and Heidelberg University for hospitality,
as well as the ERC MALADY and the ERC MPOES for financial support.   

\frenchspacing
\bibliographystyle{plain}

\end{document}